\documentclass[a4paper,11pt]{article}
\usepackage{amsmath, amssymb, amsthm,complexity}
\usepackage[ruled,vlined,linesnumbered]{algorithm2e}
\usepackage{algorithmic}
\usepackage{comment}

\usepackage{thmtools}
\usepackage{thm-restate}
\usepackage{hyperref}

\usepackage{cleveref}
\usepackage{tikz}
\usepackage{amsmath}
\usetikzlibrary{shapes}
\usetikzlibrary{plotmarks}

\usetikzlibrary{arrows,calc,shapes,decorations.pathreplacing}

\usepackage{wrapfig}
\usepackage{lscape}
\usepackage{rotating}
\usepackage{epstopdf}

\usepackage{microtype}

\usepackage{authblk}

\def\margin{2.9cm}
\usepackage[left=\margin,right=\margin,top=\margin,bottom=\margin]{geometry}

\title{Parameterized Complexity of Biclique Contraction and Balanced Biclique Contraction}

\date{}

\author[1]{R. Krithika}
\author[2]{V. K. Kutty Malu}
\author[3]{Roohani Sharma}
\author[4]{Prafullkumar Tale}

 \affil[1]{Indian Institute of Technology Palakkad, Palakkad, India\\\texttt{krithika@iitpkd.ac.in}}
 \affil[2]{Indian Institute of Technology Palakkad, Palakkad, India\\\texttt{112104002@smail.iitpkd.ac.in}}
 \affil[3]{Institute for Basic Science, South Korea\\ \texttt{roohani@ibs.re.kr}}
 \affil[4]{Indian Institute of Science Education and Research Pune, Pune, India\\ \texttt{prafullkumar@iiserpune.ac.in}}

\usepackage{listings}
\usepackage{amsmath}
\usepackage{amssymb}
\usepackage{graphicx}
\usepackage{complexity} 
\usepackage{enumitem}
\usepackage{enumitem}
\setlist{nosep}

\usepackage{color}

\newcommand{\rank}{\textsf{sf}}

\newcommand{\calF}{\mathcal{F}}
\newcommand{\calG}{\mathcal{G}}
\newcommand{\calH}{{\mathcal H}}
\newcommand{\calO}{\ensuremath{{\mathcal O}}}

\newcommand{\calS}{\mathcal{S}}

\newcommand{\calW}{\mathcal{W}}

\newcommand{\true}{\texttt{True}}
\newcommand{\false}{\texttt{False}}

\newcommand{\ETH}{\textsf{ETH}}
\newcommand{\hard}{\textsf{hard}}
\newcommand{\para}{\textsf{para}}

\newcommand{\yes}{\textsc{Yes}}
\newcommand{\no}{\textsc{No}}

\newtheorem{theorem}{Theorem}
\newtheorem{lemma}[theorem]{Lemma}

\newtheorem{observation}[theorem]{Observation}
\newtheorem{proposition}[theorem]{Proposition}
\newtheorem{reduction-rule}{Reduction Rule}
\newtheorem{branching-rule}{Branching Rule}
\newtheorem{branching-rule1}{Branching Rule}
\newtheorem{reduction rule}{Reduction Rule}
\newtheorem{preprocessing rule}{Preprocessing Rule}
\newtheorem{branching rule}{Branching Rule}
\newtheorem{marking-scheme}{Marking Scheme}
\newtheorem{definition}[theorem]{Definition}

\newcommand{\defproblem}[3]{
  \vspace{1mm}
\noindent\fbox{
  \begin{minipage}{0.96\textwidth}
  \begin{tabular*}{\textwidth}{@{\extracolsep{\fill}}lr} #1 \\ \end{tabular*}
  {\bf{Input:}} #2  \\
  {\bf{Question:}} #3
  \end{minipage}
  }
  \vspace{1mm}
}

\usepackage{hyperref}
\usepackage{thmtools}
\usepackage{thm-restate}
\usepackage[symbol]{footmisc}

\begin{document}

\maketitle

\begin{abstract}

 A bipartite graph is called a biclique if it is a complete bipartite graph and a biclique is called a balanced biclique if it has equal number of vertices in both parts of its bipartition. 
In this work, we initiate the complexity study of \textsc{Biclique Contraction} 
and \textsc{Balanced Biclique Contraction}.
In these problems, given as input a graph $G$ and an integer $k$, the objective is to determine whether one can contract at most $k$ 
edges in $G$ to obtain a biclique and a balanced biclique, respectively.
We first prove that these problems are \NP-complete\ 
even when the input graph is bipartite. 
Next, we study the parameterized complexity of these problems and
show that they admit single exponential-time \FPT\ algorithms when 
parameterized by the number $k$ of edge contractions.
Then, we show that \textsc{Balanced Biclique Contraction} admits a quadratic 
vertex kernel while \textsc{Biclique Contraction} does not admit 
any polynomial compression 
(or kernel) unless $\NP \subseteq \coNP/\poly$. 
We also give faster \FPT\ algorithms for contraction to restricted bicliques. 

\end{abstract}

\section{Introduction}

Graph modification problems have been extensively studied in theoretical computer science for their expressive power. The input of a typical graph modification problem consists of a graph $G$ and a positive integer $k$, and the objective is to make at most $k$ modifications to $G$ so that the resulting graph belongs to some specific family $\calF$ of graphs. \textsc{$\calF$-Contraction} problems refer to the variant where the only permitted modifications are edge contractions. Watanabe et al.~\cite{WatanabeAN81} and Asano and Hirata~\cite{AsanoH83} proved that if $\calF$ is closed under edge contractions and satisfies certain specific properties, then \textsc{$\calF$-Contraction} is \NP-complete. Brouwer and Veldman \cite{BrouwerV87} proved that \textsc{$\calF$-Contraction} is \NP-complete\ even when $\mathcal{F}$ is a singleton set consisting of a small graph like a cycle or a path on four vertices. Note that \textsc{$\mathcal{F}$-Vertex Deletion} and \textsc{$\mathcal{F}$-Edge Deletion} (the variants of the modification problems where the only allowed edits are vertex deletions and edge deletions) are trivially solvable when $\mathcal{F}$ is a fixed singleton set. This is one of the many examples that suggest that $\mathcal{F}$-\textsc{Contraction} problems are harder than the analogous vertex and edge deletion counterparts. 

The study of graph modification problems in parameterized complexity led to the discovery of several important techniques in the field and one may argue that it played a central role in the growth of the area  \cite{BliznetsFPP15, BliznetsFPP18, Cao16, Cao17, CaoM15, CaoM16, CrespelleDFG20, DrangeDLS22, DrangeFPV14, DrangeP18, FominKPPV14, FominV13}.  
The contrast in the apparent difficulty of contraction problems when compared to their vertex/edge deletion variants is evident even in this realm. A natural parameter for graph modification problems is the number $k$ of allowed modifications. For a family $\calH$ of graphs, we say that $G$ is $\calH$‐free if for every graph $H \in \calH$, $G$ does not contain $H$ as an induced subgraph. A graph family $\calF$ \emph{admits a forbidden set characterization} if there exists a collection $\calH$ of graphs such that $G \in \calF$ if and only $G$ is $\calH$-free. A result by Cai~\cite{Cai96} states that if $\calF$ is hereditary and admits a finite forbidden set characterization,  then the problem of modification to $\calF$ using any combination of vertex deletions, edge deletions, and edge additions, admits a single exponential-time \FPT\ algorithm. However, Cai and Guo~\cite{CaiG13} and Lokshtanov et al.~\cite{LokshtanovMS13} proved that $\calF$-\textsc{Contraction} is $\W[2]$-hard\ even when $\calF$ admits a finite forbidden set characterization. One of the intuitive reasons for this intractability is that the classical \emph{branching technique} that works for vertex deletion and edge deletion/addition variants does not straightaway work for contractions. Recently, Chakraborty and Sandeep \cite{ChakrabortyS23} studied the problem of contracting to an $\calH$‐free graph where $\calH$ is a singleton set and showed tractability and intractability results for various choices of $\calH$.

In spite of the inherent difficulty, \FPT\ algorithms for $\calF$-\textsc{Contraction} for several graph classes $\calF$ are known. If every graph in $\calF$ has bounded treewidth as in the case of $\calF$ being paths, trees, or cactus graphs, one may use Courcelle's theorem to show the existence of \FPT\ algorithms for $\calF$-\textsc{Contraction} (see~\cite[Chapter 7]{CyganFKLMPPS15} for related definitions and arguments). For other cases of $\calF$-\textsc{Contraction}, \FPT\ algorithms have been obtained using problem-specific techniques and arguments that typically involve deep insights into the structure of $\calF$. The first \FPT\ algorithm for \textsc{Bipartite Contraction} involves an interesting combination of techniques like iterative compression, important separators, and irrelevant vertices~\cite{HeggernesHLP13}, and the improved algorithm for the problem involves non-trivial applications of important separators~\cite{GuillemotM13}. \textsc{Planar Contraction} was shown to be \FPT~\cite{GolovachHP13} using the irrelevant vertex technique combined with an application of Courcelle’s theorem. 

On the other hand, \textsc{Clique Contraction} admits a relatively simpler \FPT\ algorithm running in $\calO^*(2^{\mathcal{O}(k \log k)})$ \footnote[8]{$\calO^*(\cdot)$ notation suppresses the factors that are polynomial in the input size.} time \cite{CaiG13}. This algorithm relies on the observation that if one can obtain a clique from a graph $G$ by contracting $k$ edges, then one can obtain a clique from $G$ by deleting at most $2k$ vertices (that are endpoints of the contracted edges). That is, if $(G, k)$ is a \yes-instance of \textsc{Clique Contraction}, then $V(G)$ can be partitioned into sets $X$ and $Y$ such that the cardinality of $X$ is at most $2k$ and $Y$ induces a clique. One can find such a partition (if it exists) in \FPT\ time. Then the algorithm guesses the solution edges in $E(X) \cup E(X,Y)$ and proceeds with branching. It is surprising that this simple algorithm is optimal under the Exponential-time Hypothesis (\ETH)~\cite{FominLMSZ21}. This made us wonder if such results hold for the closely related problem of \textsc{Biclique Contraction}.

\defproblem{\textsc{Biclique Contraction}}{A graph $G$ and an integer $k$.}{Can we contract at most $k$ edges in $G$ to obtain a biclique?}

\noindent We call the variant of \textsc{Biclique Contraction} where we require the resultant graph to be a balanced biclique as \textsc{Balanced Biclique Contraction}. To our pleasant surprise, we are able to show that both the problems admit simple single exponential-time \FPT\ algorithms and that neither of the problems admit subexponential-time \FPT\ algorithms.

To the best of our knowledge, even the classical complexity of \textsc{Biclique Contraction} and \textsc{Balanced Biclique Contraction} is not known in the literature. As bicliques and balanced bicliques are not closed under edge contractions, one cannot use the results by Asano and Hirata~\cite{AsanoH83} to show the \NP-hardness of these problems.
Ito et al.~\cite[Theorem~2]{ItoKPT11} showed that for each $p, q \ge 2$,  $\{K_{p,q}\}$-\textsc{Contraction} is \NP-complete where $K_{p,q}$ denotes the biclique with $p$ vertices in one part of the bipartition and $q$ vertices in the other.  This does not imply the \NP-hardness of 
{\sc Biclique Contraction} and {\sc Balanced Biclique Contraction}. Indeed,  if $G$ is a \yes-instance of $\{K_{p,q}\}$-\textsc{Contraction}, 
then for $k = |V(G)| - (p+q)$, $(G, k)$ is a \yes-instance of \textsc{Biclique Contraction}, however, the converse is not necessarily true.  A similar argument holds for \textsc{Balanced Biclique Contraction}.

The vertex-deletion variant of the problems of deleting $k$ vertices to get a biclique or a balanced biclique has received considerable attention in the literature. To be consistent with our terminology, we call these problems as \textsc{Biclique Vertex Deletion} and \textsc{Balanced Biclique Vertex Deletion}. \textsc{Biclique Vertex Deletion} is polynomial-time solvable on bipartite graphs but \NP-complete\ in general \cite[Problem GT24]{GareyJ79}. However, \textsc{Balanced Biclique Vertex Deletion} is \NP-complete\ even on bipartite graphs \cite[Problem GT24]{GareyJ79}. Our first result is that the analogous contraction problems are \NP-complete\ and they remain so even on bipartite graphs.

\begin{theorem}
\label{thm:np-hardness}
{\sc Biclique Contraction} and {\sc Balanced Biclique Contraction} are {\em \NP-complete} even when the input graph is bipartite.
\end{theorem}

We reduce \textsc{Red-Blue Dominating Set} to {\sc Biclique Contraction} and \textsc{Hypergraph $2$-Coloring} to {\sc Balanced Biclique Contraction} to show Theorem \ref{thm:np-hardness}. It is well-known that there are linear size-preserving reductions from \textsc{3-SAT} to \textsc{Red-Blue Dominating Set} and from \textsc{3-SAT} to \textsc{Hypergraph $2$-Coloring} \cite{CyganFKLMPPS15,Karp72,Schaefer78}. As the reductions used to prove Theorem~\ref{thm:np-hardness} are also linear size-preserving, it follows that {\sc Biclique Contraction} and {\sc Balanced Biclique Contraction} do not admit algorithms running in $\calO^*(2^{o(n)})$ time (and hence do not admit $\calO^*(2^{o(k)})$ time algorithms) assuming \ETH. However, as mentioned earlier, we show that both problems admit single exponential-time \FPT\ algorithms. 

\begin{theorem}
\label{thm:fpt-algo}
{\sc Biclique Contraction} and {\sc Balanced Biclique Contraction} can be solved in $\mathcal{O}^*(25.904^k)$ time.
\end{theorem}

The only other known cases when $\calF$-\textsc{Contraction} admits a single exponential-time \FPT\ algorithm is when $\calF$ is the collection of paths~\cite{HeggernesHLLP14}, trees~\cite{HeggernesHLLP14},  cactus~\cite{KrithikaMT18} or grids~\cite{SaurabhST20}. Among these results, the algorithm for paths is relatively simple, the one for trees is obtained using a non-trivial application of the color coding technique, and the last two use relatively more technical problem-specific arguments. The techniques used in our algorithms include an \FPT\ algorithm for \textsc{2-Cluster Vertex Deletion} as a subroutine, rephrasing the contraction problem as a partition problem with some properties, guessing certain vertices to be in appropriate parts of this partition, preprocessing the graph based on the guess and a branching rule.

{\sc ${\cal F}$-Contraction} for most choices of $\mathcal{F}$ mentioned here (except paths and grids) do not admit polynomial kernels. 
While we show {\sc Balanced Biclique Contraction} to be an exception, it turns out that {\sc Biclique Contraction} is not. 

\begin{theorem}
\label{thm:kernels}
{\sc Balanced Biclique Contraction} admits a kernel with $\mathcal{O}(k^2)$  vertices.
However,  {\sc Biclique Contraction} does not admit any polynomial compression or kernel unless \NP\ $\subseteq$ \coNP/poly.
\end{theorem}

As mentioned earlier, it is known that for each $p, q \ge 2$,  $\{K_{p,q}\}$-\textsc{Contraction} is \NP-complete. Observe that there is an easy brute-force algorithm (mentioned in Section~\ref{sec:preliminaries}) that solves this problem in $\mathcal{O}^*(2^k)$ time for each pair of integers $p,q \geq 2$. We add to this line of study by studying the complexity of $K_{p,*}$-\textsc{Contraction} where the objective is to determine if the input graph can be contracted to a biclique with exactly $p$ vertices in one part using at most $k$ edge contractions. Note that $p$ is fixed as part of the problem definition. 

\begin{theorem}
\label{thm:np-hard-kcq-correctness}
$K_{p, *}$-\textsc{Contraction} is \NP-complete, for each $p \in \mathbb{N}^+$.  
\end{theorem}

\begin{theorem}
\label{thm:fpt-algo-c-witness-set-in-left}
 $K_{1, *}$-\textsc{Contraction} can be solved in $\mathcal{O}^*(2^k)$ time and $K_{p, *}$-\textsc{Contraction} can be solved in $\mathcal{O}^*(3^k)$ time for each $p \geq 2$.
\end{theorem}

\section{Preliminaries}
\label{sec:preliminaries}

For details on parameterized algorithms, we refer to standard books in the area \cite{CyganFKLMPPS15, FominLSZ19}. For a positive integer $q$, we denote the set $\{1, 2, \dots, q\}$ by $[q]$. A partition of a set $S$ into disjoint sets $S_1,\ldots,S_\ell$ is denoted as $\langle S_1,\ldots,S_\ell \rangle$. An {\em ordered partition} is one where the parts are ordered. 

For an undirected graph $G$, $V(G)$ and $E(G)$ denote its set of vertices and edges, respectively. 
The {\em size} of a graph is the number of edges in it. 
A graph is {\em non-trivial} if it has at least one edge. 
We denote an edge with endpoints $u$ and $v$ as $uv$. 
Two vertices $u$ and $v$ in $V(G)$ are \emph{adjacent} if $uv$ is an edge in $G$.
The \emph{neighbourhood} of a vertex $v$, denoted by $N_G(v)$, is the set of vertices adjacent to $v$. The {\em degree} of a vertex is the size of its neighbourhood. 
A vertex $u$ is a \emph{pendant vertex} if its degree is one. 
We omit the subscript in the notation for the neighbourhood if the graph under consideration is clear. 
Subdividing an edge $uv$ results in its deletion followed by the addition of a new vertex adjacent to $u$ and $v$. 
For a subset $S \subseteq V(G)$ of vertices, $N[S]=\bigcup_{v \in S} N(v) \cup \{v\}$ and $N(S) = N[S] \setminus S$.
We denote the subgraph of $G$ induced on the set $S$ by $G[S]$. 
A subset $S\subseteq V(G)$ is said to be a \emph{connected set} if $G[S]$ is connected. 
For two subsets $S_1, S_2  \subseteq V(G)$, $E_G(S_1, S_2)$ denotes the set of edges with one endpoint in $S_1$ and the other endpoint in $S_2$. With a slight abuse of notation, for a set $S \subseteq V(G)$, we use $E_G(S)$ to denote $E_G(S, S)$. 
We say that $S_1, S_2$ are adjacent (or that the graphs $G[S_1]$ and $G[S_2]$ are adjacent) if $E_G(S_1, S_2) \neq \emptyset$. 
We omit the subscript in these notations if the graph under consideration is clear. %
The {\em disjoint union} of graphs $G$ and $H$, denoted by $G+H$, is the graph with vertex set $V(H) \cup V(G)$ and edge set $E(H) \cup E(G)$ where $V(G)$ and $V(H)$ are renamed (if necessary) such that $V(G) \cap V(H) =\emptyset$. 
For a graph $G$, $\overline{G}$ denotes its complement. 

A {\em path} is a sequence of distinct vertices in which any two consecutive vertices are adjacent. 
A {\em cycle} is a path in which the first and last vertices are adjacent.  
A graph is {\em connected} if there is a path between every pair of distinct vertices. 
A {\em component} is a maximal connected subgraph of a graph. 
A \emph{spanning tree} of a connected graph is a connected acyclic subgraph that has all the vertices of the graph.
A \emph{spanning forest} of a graph is a collection of spanning trees of its components.
A set of vertices is called a \emph{clique} if any two vertices in it are adjacent.
A {\em complete graph} is one whose vertex set is a clique and a complete graph on $n$ vertices is denoted by $K_n$.
A set of vertices $S$ is said to be an \emph{independent set} if no two vertices in $S$ are adjacent. Throughout this paper, the input graph will be assumed to be connected. 

A graph $G$ is \emph{bipartite} if its vertex set can be partitioned into 
two independent sets $X$, $Y$, and $\langle X, Y \rangle$ is called a {\em bipartition} of $G$. It it well-known that a graph $G$ is bipartite if and only if $G$ has no odd cycles. 
A bipartite graph $G$ is a \emph{biclique} if it has a bipartition $\langle X, Y \rangle$ such that every vertex in $X$ is adjacent to every vertex in $Y$. 
Observe that an edgeless graph is also a biclique. 
A biclique $G$ with bipartition $\langle X, Y \rangle$ is a \emph{balanced biclique} if $|X| = |Y|$. 
Observe that if a biclique is non-trivial, then it is connected. 
The following characterization of bicliques is easy to verify.

\begin{lemma}
\label{lemma:biclique-forbidden-char}
A graph $G$ is a biclique if and only if $G$ is $\{K_3,K_1+K_2\}$-free. 
\end{lemma}

\begin{proof}
$(\Rightarrow)$ If $G$ contains $K_3$ as an induced subgraph then 
$G$ is not bipartite and hence not a biclique as well.  
Suppose $G$ is a biclique with bipartition $\langle X, Y \rangle$ and 
has $K_1+K_2$ as an induced subgraph. 
Let $u,v_1,v_2$ be three vertices in $G$ such that $v_1v_2 \in E(G)$ 
and $uv_1,uv_2 \notin E(G)$. 
Without loss of generality, let $X$ be the part that contains $v_1$. 
Then, it follows that $v_2 \in Y$. 
However, $u \in X$ leads to a contradiction as $uv_2 \notin E(G)$
 and $u \in Y$ also leads to a contradiction as $uv_1 \notin E(G)$. 

$(\Leftarrow)$ Conversely,  suppose $G$ is $\{K_3,K_1+K_2\}$-free. If $G$ is not bipartite, then $G$ has an odd cycle and
let $C=(u_1,u_2,\ldots,u_r)$ denote the shortest odd cycle. 
Observe that no two non-consecutive vertices are adjacent, as $C$ is a 
shortest odd cycle. 
Then $C$ has at least $5$ vertices (as $G$ is $K_3$-free) and the subgraph 
induced on $\{u_1,u_2,u_4\}$ is $K_1+K_2$. 
Therefore,  we conclude that $G$ has no odd cycle and is hence a bipartite graph. 
Let $\langle X, Y \rangle$ denote a bipartition of $G$. 
If $G$ is edgeless, then $G$ is vacuously a biclique. 
Subsequently, let us assume that $G$ is non-trivial.  
If $G$ is disconnected, then it has at least two components, 
one of which is non-trivial. 
Then, the endpoints of an edge in such a non-trivial component 
and a vertex from another component induces a $K_1+K_2$. 
Therefore, it follows that $G$ is connected. 
Suppose there exist a pair of vertices $x \in X$, $y \in Y$ such that 
$xy \notin  E(G)$. 
Let $z \in Y$ be a neighbour of $x$.  
Such a vertex $z \in Y$ exists as $G$ is connected. 
Then, the subgraph induced on $\{x,y,z\}$ is $K_1+K_2$, which leads to a contradiction.
\end{proof}

\subsection{Edge Contractions and Graph Contractability}
The {\em contraction} of an edge $e = uv$ in $G$ results in the addition of a new vertex that is adjacent to the vertices that are adjacent to either $u$ or $v$ and the subsequent deletion of the vertices $u$ and $v$. The resulting graph is denoted by $G/e$. 
For a graph $G$ and an edge $e = uv$, we formally define $G/e$ as $V(G/e) = (V(G) \cup \{w\}) \backslash \{u, v\}$ and $E(G/e) = \{xy \mid x,y \in V(G) \setminus \{u, v\}, xy \in E(G)\} \cup \{wx \mid x \in N_G(u) \cup N_G(v)\}$ where $w$ is a new vertex. 
This process does not introduce any self-loops or parallel edges. 
For a subset $F \subseteq E(G)$, the graph $G/F$ denotes the graph 
obtained from $G$ by contracting (in an arbitrary order) all the edges in $F$.

A graph $G$ is said to be \emph{contractible} to a graph $H$ if there is a function $\psi: V(G) \rightarrow V(H)$ (and we say that $G$ is contractible to $H$ via $\psi$) such that the following properties hold.
\begin{itemize}
\item For any vertex $h \in V(H)$, $\psi^{-1}(h)$ is non-empty and connected.
\item For any two vertices $h, h’ \in V(H)$, $hh’ \in E(H)$ if and only if $E(\psi^{-1}(h), \psi^{-1}(h')) \neq \emptyset$.
\end{itemize}
\noindent For a vertex $h$ in $H$, the set $\psi^{-1}(h)$ is called a \emph{witness set} associated with or corresponding to $h$.
For a fixed $\psi$, we define the $H$-\emph{witness structure} of $G$, denoted by $\mathcal{W}$, as the collection of all witness sets.
Formally, $\mathcal{W}=\{\psi^{-1}(h) \mid h \in V(H)\}$.
Note that a witness structure $\mathcal{W}$ partitions $V(G)$.
If a witness set contains more than one vertex, then we call it a 
\emph{big} witness set; otherwise, we call it a \emph{small} witness set or \emph{singleton} witness set.
Let $F \subseteq E(G)$ be the collection of edges of some spanning tree of $G[W]$ for each witness set $W \in \mathcal{W}$. Note that any spanning tree of the graph induced on a singleton witness set is edgeless. 
Now, it is sufficient to contract edges in $F$ to obtain $H$ from $G$, i.e., $G/F = H$.
Hence, we say that $F$ is a \emph{solution} associated with the function $\psi$ and the witness structure $\mathcal{W}$.
We say $G$ is \emph{$k$-contractible} to $H$ if there exists a subset $F \subseteq E(G)$ such that $|F| \le k$ and $G/F = H$. 
Observe that in the $H$-witness structure of $G$ corresponding to $F$, there are at most $|F|$ big witness sets and the total number of vertices in
big witness sets is upper bounded by $2|F|$. 

We view a biclique witness structure of $G$ as a partition of $V(G)$ into two parts with certain properties. 
For a subset $X \subseteq V(G)$, let $\rank_G(X)$ denote the number of edges in a spanning forest of $G[X]$. We drop the subscript in the notation for $\rank_G(X)$ if the graph under consideration is unambiguously clear. 

\begin{definition}[$k$-Constrained Valid Partition]
\label{def:k-contr-valid-partition}
For a graph $G$, a partition $\langle L,R \rangle$ of $V(G)$ into two parts is called a $k$-constrained valid partition if the following properties hold.
\begin{enumerate}
\item $\rank(L) + \rank(R) \le k$.
\item Every component of $G[L]$ is adjacent to every component of $G[R]$.
\end{enumerate}
\end{definition}
We have the following observation on \yes-instances of \textsc{Biclique Contraction}. 

\begin{lemma}
\label{lem:bc-part}
$(G,k)$ is a \yes-instance of \textsc{Biclique Contraction} if and only if $V(G)$ has a $k$-constrained valid partition.
\end{lemma}

\begin{proof}
$(\Rightarrow)$ 
Suppose $G$ is $k$-contractible to the biclique $H$ with bipartiton $\langle X, Y \rangle$ via $\psi$.  
Let $\mathcal{W}$ denote the $H$-\emph{witness structure} of $G$. 
Define subsets  $\mathcal{W_L}$ and $\mathcal{W_R}$ of $\mathcal{W}$ as follows: $\mathcal{W_L}=\{\psi^{-1}(h) \mid h \in V(X)\}$ and $\mathcal{W_R}=\{\psi^{-1}(h) \mid h \in V(Y)\}$. Let $L$ denote the collection of vertices in the witness sets in $\mathcal{W}_L$ and $R$ denote the collection of vertices in the witness sets in $\mathcal{W}_R$. 
It is clear that $\langle L, R \rangle$ is a partition of $V(G)$. 
Let $E_1$ be the collection of edges that are in some spanning forest of $G[L]$ and $E_2$ be the collection of edges that are in some spanning forest of $G[R]$. 
Then, contracting $E_1$ in $G[L]$ results in $X$ and contracting $E_2$ in $G[R]$ results in $Y$. Also, $E_1 \cup E_2 = E$ is a solution corresponding to $\mathcal{W}$. Then, as $|E| \leq k$, the first condition holds. 
As $H$ is a biclique, every vertex of $X$ is adjacent to every vertex of $Y$.
Therefore, for every vertex $h \in X$ and every vertex $h' \in Y$, $E(\psi^{-1}(h), \psi^{-1}(h'))$ is not empty. In other words, for every component $C_L$ of $G[L]$ and every component $C_R$ of $G[R]$, there exists vertices $x \in V(C_L)$ and $y \in V(C_R)$ such that $xy \in E(G)$. Hence, the second condition holds. 

$(\Leftarrow)$ Suppose $\langle L,R \rangle$ is a $k$-constrained valid  partition of $V(G)$. Let $E_1$ be the collection of edges of some spanning forest of $G[L]$ and $E_2$ be the collection of edges of some spanning forest of $G[R]$. 
Observe that $G[L]/ E_1 = X$ and $G[L]/ E_2 = Y$ are independent sets.  
Further, every vertex in $X$ is adjacent to every vertex in $Y$ due to the second property. 
Therefore $G/(E_1 \cup E_2) = H$ is a biclique with bipartiton 
$\langle X, Y \rangle$. 
Further, as $|E_1| = k_\ell$ and $|E_2| = k_r$ with $k_\ell + k_r \leq k$ due to the first property, it follows that
 $G$ is $k$-contractible to $H$. 
\end{proof}

Suppose $G$ is $k$-contractible to the biclique $H$ with $\mathcal{W}$ being the $H$-witness structure of $G$ and $\langle L, R \rangle$ is a $k$-constrained valid partition of $V(G)$ obtained from $\mathcal{W}$ as described in the proof of Lemma \ref{lem:bc-part}. Then, observe that the singleton witness sets in $\mathcal{W}$ correspond to the trivial components of $G[L]$ and $G[R]$ and vice-versa. A component is said to be trivial component if it has only a single vertex. We use this equivalence and the interchangeability of witness structures and constrained valid partitions throughout the paper. 

In order to prove a result on \yes-instances of \textsc{Balanced Biclique Contraction} analogous to Lemma \ref{lem:bc-part}, we introduce the following definition.

\begin{definition}[$k$-Constrained Valid Balanced Partition]
For a graph $G$, a partition $\langle L,R \rangle$ of $V(G)$ into two parts is called a $k$-constrained valid balanced partition if $\langle L,R \rangle$ is a $k$-constrained valid partition where the number of components of $G[L]$ is equal to the number of components of $G[R]$.
\end{definition}

Now, we have the following property on \yes-instances of \textsc{Balanced Biclique Contraction}. 

\begin{lemma}
\label{lem:bbc-part}
$(G,k)$ is a \yes-instance of \textsc{Balanced Biclique Contraction} if and only if $V(G)$ has a $k$-constrained valid balanced partition.
\end{lemma}

Observe that there is an easy brute-force algorithm that solves $K_{p,q}$-\textsc{Contraction} in $\mathcal{O}^*(2^k)$ time for each pair $p,q \geq 2$. Consider an instance $(G,k)$. If $n > k+p+q$ then $(G,k)$ is a trivial \no-instance. Otherwise, $n \leq k+p+q$. In $\mathcal{O}{(2^{k+p+q})}$ (which is $\mathcal{O}^*(2^k)$) time, we check if $V(G)$ has a $k$-constrained valid partition $\langle L,R \rangle$ with $G[L]$ having exactly $p$ components and $G[R]$ having exactly $q$ components.

We end this section by mentioning an important definition and observation. A set $Z \subseteq V(G)$ is called a {\em biclique modulator} if $G-Z$ is a biclique. If $G$ is $k$-contractible to a (balanced) biclique $H$, then there are at most $2k$ vertices that are in big witness sets of an $H$-witness structure. This leads to the following observation.

\begin{observation}
\label{obs:bc-vertex-sol}
If $G$ is $k$-contractible to a (balanced) biclique, then $G$ has a biclique modulator of size at most $2k$. 
\end{observation}

\section{NP-Completeness Results}
\label{sec:npc}

In this section, we prove that \textsc{Biclique Contraction} and \textsc{Balanced Biclique Contraction} are \NP-complete\ even when restricted to bipartite graphs. 

\subsection{Biclique Contraction}
\label{subsec:np-hard-biclique}

A set $X$ is said to {\em dominate} a set $Y$ if $Y \subseteq N(X)$. 
We show the \NP-hardness of {\sc Biclique Contraction} by giving a polynomial-time reduction from \textsc{Red-Blue Dominating Set}.   
In the \textsc{Red-Blue Dominating Set} problem, given a bipartite graph $G$ with bipartition $\langle R, B \rangle$ and an integer $\kappa$, the objective is to find a set $S \subseteq R$ of size at most $\kappa$ that dominates $B$.
It is well-known and easy to verify that \textsc{Red-Blue Dominating Set} is equivalent to \textsc{Set Cover} \cite[Problem SP5]{GareyJ79} and is therefore \NP-hard \cite{DomLS14}.

\begin{lemma}
\label{lem:RBDS-to-biclique-contraction}
There is a polynomial-time reduction that takes as input an instance $(G, R, B, \kappa)$ of {\sc Red Blue Dominating Set} and returns an equivalent instance $(H, k)$ of {\sc Biclique Contraction} such that $H$ is bipartite, $|V(H)| = \calO(|V(G)|)$ and $k = |B| + \kappa$. 
\end{lemma}
\begin{proof}
Consider an instance $(G, R, B, \kappa)$ of \textsc{Red-Blue Dominating Set}.
Without loss of generality, assume $|R|,|B| >\kappa$ and that every vertex $b \in B$ is adjacent to at least two vertices in $R$. 
We construct an instance $(H,k)$ of {\sc Biclique Contraction} where $k=\kappa+|B|$ and $H$ is obtained from $G$ as follows.

\begin{figure}[t]
\begin{center}
\includegraphics[scale=0.2]{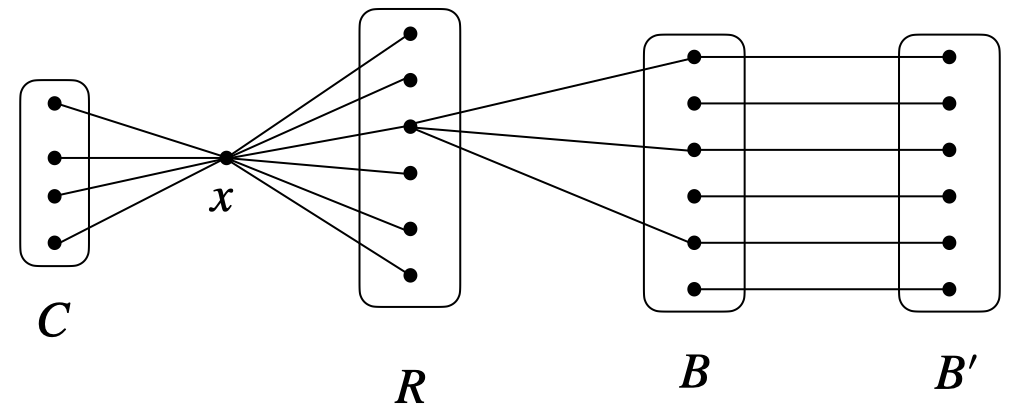}
\caption{The graph $H$ in the reduction from \textsc{Red Blue Dominating Set} 
to \textsc{Biclique Contraction} where edges between $R$ and $B$ are the same as in $G$. 
The vertex $x$ is adjacent to every vertex in $R \cup C$, and each vertex in $B$ is adjacent to exactly one vertex in $B'$.
\label{fig:np-hard-biclique}}
\end{center}
\end{figure}
\begin{itemize}
\item For every vertex $b \in B$, add a new vertex $b'$ adjacent to $b$.
Let $B'$ denote the set $\{b' \mid b \in B\}$.
\item Add a new vertex $x$ adjacent to every vertex in $R$ and add $\kappa+|B|+1$ new vertices $v_1, \ldots, v_{\kappa+|B|+1}$ adjacent to $x$.
Let $C$ denote the set $\{v_1, \ldots, v_{\kappa+|B|+1}\}$.
\end{itemize}
See Figure \ref{fig:np-hard-biclique} for an illustration. 

It is easy to verify that the reduction takes polynomial time. Further, 
 $H$ is connected and bipartite with bipartition $\langle B \cup \{x\}, R \cup C \cup B' \rangle$. 
 We show that $(G, R, B, \kappa)$ is a \yes-instance of \textsc{Red-Blue Dominating Set} if and only if $(H,k)$ is a \yes-instance of {\sc Biclique Contraction} where $k=\kappa+|B|$.

$(\Rightarrow)$
Let $(G, R, B, \kappa)$ is a \yes-instance of \textsc{Red-Blue Dominating Set} with set $S \subseteq R$ and $|S| \leq \kappa$ such that $S$ dominates $B$.
Then, consider the partition of $V(H)$ into sets $X=S \cup B \cup \{x\}$ and $Y=C \cup B' \cup (R \setminus S)$. 
Observe that $H[X]$ is connected as $S$ dominates $B$ and $S \subseteq N(x)$. Therefore, any spanning forest, which is also a spanning tree, 
of $H[X]$ has at most $\kappa+|B|$ edges. Further, $Y$ is an independent set and since $H$ is connected, 
every vertex in $Y$ has a neighbour in $X$. Therefore, $H/F$ is a biclique where $F$ is the  set of edges of some spanning tree of $H[X]$. 

$(\Leftarrow)$
Conversely, suppose $H$ is $k$-contractible to a biclique. Let $\langle X, Y \rangle$ be a $k$-constrained valid partition of 
$V(H)$ given by Lemma~\ref{lem:bc-part}. Without loss of generality, let $x \in X$.
We first claim that at least one vertex from $C$ is in $Y$. If $C \subseteq X$, then as $C \subseteq N(x)$ and $|C|=k+1$, 
any spanning forest of $H[X]$ has at least $k+1$ edges, leading to a contradiction. Let $c \in C$ be a vertex in $Y$.
Next, we claim that $H[X]$ is connected. Suppose $H[X]$ has a component $\widehat{H}$ not containing $x$.
Then, $\widehat{H}$ has no vertex from $C$ as $C \subseteq N(x)$ and no vertex from $R$ as $R \subseteq N(x)$. 
It follows that $V(\widehat{H}) \subseteq B \cup B'$.  Then,  the component of $H[Y]$ containing
$c$ has no other vertex since $N(c)=\{x\}$. This implies that no vertex in the component of $H[Y]$ containing 
$c$ is adjacent to a vertex in $\widehat{H},$ leading to a contradiction.
Now, we argue that $\langle X, Y \rangle$ can be transformed into another $k$-constrained valid partition of $V(H)$ with $C \subseteq Y$. If there is a vertex $d \in X \cap C$, then as $N(d)=\{x\}$, $H[X \setminus \{d\}]$ is connected and every component of $H[Y \cup \{d\}]$ (including the trivial component in which $d$ is in) is adjacent to $H[X \setminus \{d\}]$, $\langle X \setminus C, Y \cup C \rangle$ is the required $k$-constrained valid partition of $V(H)$. 

Now, we transform $\langle X, Y \rangle$ into another $k$-constrained valid partition 
of $V(H)$ such that $B \subseteq X$ and $B' \subseteq Y$. 
Suppose there is a vertex $b \in B \cap Y$.
Let $\widehat{H}$ be the component of $H[Y]$ containing $b$.
Then, as the only neighbour of $b'$ is $b$ and $x\in X$, it follows that $b' \in V(\widehat{H})$.
As $H[X]$ is adjacent to $\widehat{H}$, there is a vertex $r \in R$ that is adjacent to $b$ such that $r \in V(\widehat{H})$ or $r \in X$. 
In the former case, we move $B \cap V(\widehat{H})$ and $R \cap V(\widehat{H})$ to $X$ and in the latter case, we move $B \cap V(\widehat{H})$ to $X$.
In both cases, $X$ remains connected, and it is easy to verify that the resulting partition is a $k$-constrained valid partition of $V(H)$. 

Subsequently, we may assume $B \subseteq X$. 
As $H[X] \setminus B'$ is also connected and 
$N(B') \subseteq B$, by moving vertices from $X \cap B'$ to $Y$, 
we get another partition of $V(H)$ that is a $k$-constrained valid partition. 
Once we achieve $B \subseteq X$, we may safely move vertices of $B'$ 
from $X$ (if any) to $Y$. 
Hence, we may now assume $B \subseteq X$ and $B' \subseteq Y$. 

At this point, we have $B \cup \{x\} \subseteq X$ and $B' \cap X=\emptyset$.
Also, as $N(x) \cap B=\emptyset$, it follows that for each vertex $b \in B$, there is a vertex $r \in R \cap X$ in order for $H[X]$ to be connected. 
As any spanning tree of $H[X]$ has at most $k=\kappa+|B|$ edges, it follows that the set $R \cap X$ has at most $\kappa$ vertices.
Equivalently, $R \cap X$ is a set of at most $\kappa$ vertices that dominates $B$.
\end{proof}

It is easy to verify that {\sc Biclique Contraction} is in \NP.
This fact, along with Lemma \ref{lem:RBDS-to-biclique-contraction} establishes the first part of Theorem \ref{thm:np-hardness}.

\subsection{Balanced Biclique Contraction}
\label{subsec:np-hard-bal-biclique}
We show the \NP-hardness of \textsc{Balanced Biclique Contraction} by a reduction from \textsc{Hypergraph $2$-Coloring}. 
A {\em hypergraph} $\calG$ is a pair $(V,\mathcal{S})$ where $V$ is a finite set of vertices (denoted as $V(\calG)$) and $\calS \subseteq 2^{V}$ is a finite collection of subsets of $V$ called {\em hyperedges}. 
In the \textsc{Hypergraph $2$-Coloring} problem, the input is a hypergraph $\calG$ and the objective is to determine if there is a $2$-coloring $\phi: V(\calG) \mapsto \{1, 2\}$ such that no hyperedge is monochromatic, i.e., a $2$-coloring in which every hyperedge has a vertex with color 1 and a vertex with color 2. 
\textsc{Hypergraph $2$-Coloring} is known to be \NP-complete~\cite[Problem SP4]{GareyJ79} and is one of the natural choices picked to show the \NP-hardness of contraction problems in the literature \cite{BrouwerV87,DabrowskiP17-IPL,FialaKP13,KrithikaST22}.


\begin{lemma}
\label{lem:np-hard-bal-biclique-correctness}
There is a polynomial-time reduction that takes as input an instance $(\calG)$
of {\sc Hypergraph $2$-Coloring}  and returns an equivalent instance 
$(G, \kappa)$ of {\sc Balanced Biclique Contraction} such that 
$G$ is bipartite, $|V(G)| = \mathcal{O}(|V(\calG)|)$ and $\kappa$ is a function of the number of vertices and the number of hyperedges in $\calG$. 
\end{lemma}
\begin{proof}
Consider an instance $(\calG=(V,\mathcal{S}))$ of {\sc Hypergraph $2$-Coloring} where $V(\calG) = \{v_1, v_2, \dots, v_N\}$ for some $N \ge 1$ and $\calS = \{S_1, S_2, \dots, S_M\}$ for some $M \ge 1$. 
Without loss of generality, assume that $\emptyset \not\in \calS$, every hyperedge contains at least two vertices and there exists an hyperedge, say $S_M$, that contains all the vertices in $V(\calG)$. 
The reduction first constructs an intermediate non-bipartite graph $H$ using the following procedure. 
\begin{itemize}
\item For every vertex $v$ in $V(\calG)$, add a vertex $v$ in $H$. 
\item For every hyperedge $S_j$ in $\calS$, add two vertices $s^{\ell}_j$ and $s^{r}_j$. Let $S^{\ell} = \{s^{\ell}_j \mid j \in [M]\}$ and $S^r = \{s^{r}_j \mid j \in [M]\}$.
\item For every $i \in [N]$ and $j \in [M]$ such that $v_i \in S_j$ add edges $v_i s^{\ell}_j$ and $v_is^{r}_j$.
\item Add vertices $L = \{\ell_1, \dots, \ell_{6M+3N-5}\}$ and $R = \{r_1, \dots, r_{6M+3N-5}\}$ such that every vertex in $L$ is adjacent to every vertex in $R$. 
\item Make every vertex in $L$ adjacent to every vertex in $S^{r}$ and every vertex in $R$ adjacent to every vertex in $S^{\ell}$.
\end{itemize}
This completes the construction of $H$.
See Figure~\ref{fig:np-hard-bal-biclique} for an illustration. 
We show that $(\calG)$ is a \yes-instance of {\sc Hypergraph $2$-Coloring} if and only if $(H, k=2M+N-2)$ is a \yes-instance of \textsc{Balanced Biclique Contraction}. The key idea behind the reduction is as follows. First, the number of vertices in $L$ and $R$ is so large that no edge in $E_H(L, R)$ can be contracted without exceeding the budget $k$. 
Hence, $H[L \cup R]$ acts as a representative subgraph of the final balanced biclique. Second, any contraction solution needs to partition $V(\calG)$ into two parts $V_1$ and
$V_2$ such that both $H[S^{\ell} \cup V_1]$ and $H[S^{r} \cup V_2]$ are connected. 
This is equivalent to finding a $2$-coloring of  $\calG$ such that
no hyperedge is monochromatic. 

\begin{figure}[t]
  \begin{minipage}[c]{0.60\textwidth}
  \includegraphics[scale=0.25]{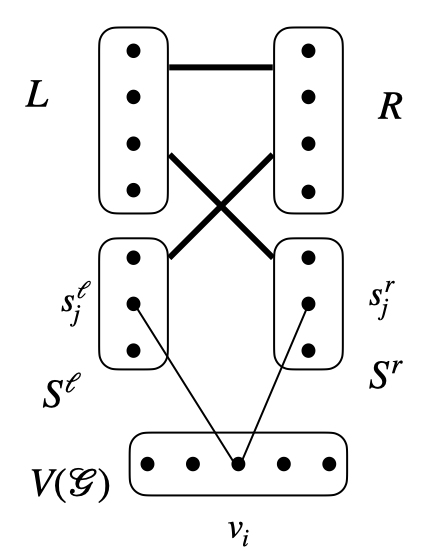}
  \end{minipage}\hfill
  \begin{minipage}[c]{0.40\textwidth}
    \caption{Construction of graph $H$ in the reduction from \textsc{Hypergraph 2-Coloring} to \textsc{Balanced Biclique Contraction} where thick lines between two sets indicate bicliques. 
    \label{fig:np-hard-bal-biclique}}
  \end{minipage}
\end{figure}

$(\Rightarrow)$ Suppose $\phi: V(\calG) \mapsto \{1, 2\}$ is a 
$2$-coloring such that for no hyperedge in $\calS$ is monochromatic. 
Using this coloring, we define a partition $\langle X,Y \rangle$ of $V(H)$
with $X= L \cup  S^{\ell} \cup V_1$ and $Y= R \cup  S^{r} \cup V_2$ 
where $V_1= \{v \mid  v\in V(\calG),  \phi(v) = 1 \}$  and 
$V_2= \{v \mid v \in V(\calG), \phi(v) = 2 \}$. 
As $\emptyset \not\in \calS$, ${S_M} = V(\calG)$, 
and no edge in $\mathcal{S}$ is monochromatic, 
$H[S^{\ell} \cup V_1]$ and $H[S^{r} \cup V_2]$ are connected. 
Moreover, $H[S^{\ell} \cup V_1]$ has a vertex 
adjacent to a vertex in $H[S^r \cup V_2]$. 
Further, as $L \cup S^{\ell}$, $L \cup V(\calG)$, $R \cup S^{r}$ and
$R \cup V(\calG)$ are independent sets in $H$, it follows that $H[X]$ 
has a spanning forest of size $M+|V_1|-1$ and $H[Y]$ has a spanning
forest of size $M+|V_2|-1$. 
Therefore, $\langle X,Y \rangle$ is a $k$-constrained valid balanced partition 
of $V(H)$ where $k = 2M + N - 2$.
From Lemma~\ref{lem:bbc-part}, we conclude that $(H, k)$ 
is a \yes-instance of \textsc{Balanced Biclique Contraction}.

$(\Leftarrow)$ Conversely, suppose $H$ is $k$-contractible to a balanced 
biclique. Recall that an edge contraction reduces the number of vertices
by exactly one. 
Then, $H$ is $k$-contractible to the balanced biclique $K_{q,q}$ 
with $q \ge 6M+3N-4=3k+2$. 
As in any $K_{q,q}$-witness structure of $H$ there are at most $2k$ 
vertices in big witness sets and $|L|=|R|=3k+1$,  there exist subsets 
$L' \subseteq L$, $R' \subseteq R$ with $|L'|,|R'| \geq k+1$ 
such that every vertex of $L' \cup R'$ is in a small witness set. 
Now let $\langle L^{*}, R^{*} \rangle$ be a $k$-constrained valid 
balanced partition of $V(H)$ given by Lemma~\ref{lem:bbc-part} 
corresponding to this witness structure. 
Then, $H[L^*]$ and $H[R^*]$ each have atleast $3k+2$ components. 
Further,  every vertex in $L'$ is in a trivial component of $H[L^*]$ and
every vertex in $R'$ is in a trivial component of $H[R^*]$. 

Next, we claim that  $R \subset R^{*}$ and $L \subset L^{*}$. 
If there exist a vertex $r \in R \cap L^{*}$, then any spanning tree of 
$H[L' \cup \{r\}]$ has at least $k+1$ edges leading to a contradiction. 
A similar argument shows that $L \subset L^{*}$. 
Now we argue that every vertex of $L$ is in a trivial component of 
$H[L^*]$ and every vertex of $R$ is in a trivial component of $H[R^*]$. 
Suppose there are two vertices $\ell_i, \ell_j \in L$ that are in 
the same component $X$ of $H[L^*]$.
As $\ell_i$ and $\ell_j$ are not adjacent, there is at least one more 
vertex $x$ in $X$ that is adjacent to $\ell_i$ or $\ell_j$. 
By the construction of $H$, $x \in R \cup S^r$. 
However, as $R \subset R^{*}$, we have $x \in S^r$. 
Then, any spanning tree of $H[\{x\} \cup L']$ has at least 
$k+1$ edges leading to a contradiction. 
Therefore, every vertex of $L$ is a trivial component of $H[L^*]$. 
A similar argument shows that every vertex of $R$ is in a trivial 
component of $H[R^*]$. This shows that vertices of $L$ form $3k+1$ components of $H[L^*]$ and vertices of $R$ form $3k+1$ components of $H[R^*]$. 

Further, as every vertex in $S^r$ is adjacent to every vertex in $L$ 
and every vertex in $S^\ell$ is adjacent to every vertex in $R$, 
we have ${S^r} \subset R^{*}$ and ${S^l} \subset L^{*}$. 
Moreover, as $L \cup V(\calG)$ and $R \cup V(\calG)$ are 
independent sets, we conclude that no vertex from $V(\calG)$ 
is in a trivial component of $H[L^{*}]$ or $H[R^{*}]$. 
Observe that ${S^\ell}$ is contained in one component of 
$H[L^{*}]$ if and only if ${S^r}$ is contained in one component 
of $H[R^{*}]$ due to Lemma~\ref{lem:bbc-part}. 
Let $\langle V_1, V_2 \rangle$ be the partition of $V(\calG)$ 
such that $V_1 = L^* \cap V(\calG)$ and $V_2= R^* \cap V(\calG)$. 
Without loss of generality, let $V_1 \neq \emptyset$. 
Suppose there are two components of $H[L^{*}]$ that contain 
vertices from $S^\ell$. 
Let $C'$ be one of these components that contains $s^{\ell}_M$. 
Then, $V_1$ is contained in $C'$. 
Let $C$ denote a component of $H[L^{*}]$ different from $C'$ 
that contains vertices from $S^\ell$. 
As $\emptyset \notin \mathcal{S}$, it follows that $C$ has a vertex 
that is adjacent to a vertex in $C'$, leading to a contradiction. 
Thus, $H[L^{*}]$ and $H[R^{*}]$ each have exactly $3k+2$ 
components implying that $H[{S^\ell} \cup V_1]$ and $H[{S^r} \cup V_2]$ 
are connected. 
Then, $\langle V_1, V_2 \rangle$ gives a $2$-Coloring for $\calG$ such that
 no edge is monochromatic. Hence, $(\calG)$ is a \yes-instance of {\sc Hypergraph $2$-Coloring}. 

Now, let $G$ be the graph obtained from $H$ by subdividing every edge
in $E_H(S^\ell,V(\calG))$. 
Clearly, $G$ is bipartite with bipartition 
$\langle L \cup S^\ell \cup V(\calG), R \cup S^r \cup Z \rangle$ 
where $Z$ is the set of new vertices added during the subdivision. 
Observe that $|Z|=\sum_{v \in V(\calG)} d(v) \geq 2M$ where $d(v)$ 
denotes the number of hyperedges that contain $v$. 
Then, $(H, k)$ is a \yes-instance of \textsc{Balanced Biclique Contraction}
if and only if $(G, k+|Z|)$ is a \yes-instance of \textsc{Balanced Biclique Contraction}. 

Suppose $H$ is $k$-contractible to the balanced biclique $K_{q,q}$. 
Then, from earlier arguments, we know that $q=3k+2$. 
Let $\langle L^{*}, R^{*} \rangle$ be a $k$-constrained valid balanced 
partition of $V(H)$  given by Lemma~\ref{lem:bbc-part}. 
Then, $\langle L^{*} \cup Z, R^{*} \rangle$ is a $(k+|Z|)$-constrained 
valid balanced partition of $V(G)$. 
By Lemma~\ref{lem:bbc-part}, $G$ is $(k+|Z|)$-contractible to a balanced biclique. 
Conversely, suppose $G$ is $(k+|Z|)$-contractible to the balanced 
biclique $K_{p,p}$. 
Let $F \subseteq E(G)$ be the set of edges of $G$ such that $G/F=K_{p,p}$. 
Then, as $|V(G)|=14M+7N-10+|Z|$, 
it follows that $2p \geq 12M+6N-8$ and so $p \geq 5$. 
As every vertex in $Z$ has degree 2 in $G$, it follows that 
no vertex of $Z$ is in a singleton witness set of a $K_{p,p}$-witness structure of $G$. 
In other words, each vertex of $Z$ has at least one edge incident on it in $F$. 
For every $z \in Z$, add one such edge to $F'$. 
Then, $|F'|=|Z|$ as no edge in $F$ has both endpoints in $Z$
since $Z$ is an independent set. Now, we have $G/F'=H$ and 
$H/(F \setminus F')=K_{p,p}$. 
Since $|F'|=|Z|$, it follows that $|F \setminus F'| \leq k$ and thus $H$ is $k$-contractible to a balanced biclique.
\end{proof}

It is easy to verify that {\sc Balanced Biclique Contraction} is in \NP. This fact, along with Lemma \ref{lem:np-hard-bal-biclique-correctness}, establishes the second part of Theorem \ref{thm:np-hardness}.

\section{FPT Algorithms}
\label{sec:fpt-algo}
In this section, we show that \textsc{Biclique Contraction} and \textsc{Balanced Biclique Contraction} can be solved in $\mathcal{O}^*(25.904^k)$ time. 
For the sake of convenience and to avoid unnecessary repetition, we first describe the algorithm for \textsc{Biclique Contraction} and later mention the (minor) changes required in order to solve the balanced variant.

H{\"{u}}ffner et al.~\cite{HuffnerKMN10} gave an $\mathcal{O}(1.4^k \cdot k^d + n^{3})$ time algorithm for the \textsc{$d$-Cluster Vertex Deletion} problem in which given a graph $G$ and integer $k$, the goal is to determine whether one can delete at most $k$ vertices such that the resultant graph is the disjoint union of at most $d$ complete graphs.
As $ G-Z$ is a biclique if and only if $\overline{G} - Z$ is the disjoint union of at most two complete graphs, we have the following result.

\begin{proposition}{\em ~\cite{HuffnerKMN10}}
\label{prop:bc-v-algo}
Given a graph $G$ and an integer $k$, there is an algorithm that either computes a biclique modulator $Z$ of size at most $2k$ or correctly determines that no such modulator exists in $\mathcal{O}^*(1.4^{|Z|})$ time.
\end{proposition}

Consider an instance $(G,k)$ of \textsc{Biclique Contraction}. 
We first determine if $G$ has a biclique modulator $Z$ of size at most $2k$ in $\mathcal{O}^*(1.4^{|Z|})$ time using Proposition \ref{prop:bc-v-algo}.  
If no such set exists, then by Observation \ref{obs:bc-vertex-sol}, we can declare that $G$ is not $k$-contractible to a biclique.  
Otherwise,  let $\langle X, Y \rangle$ denote a bipartition of $G-Z$,  and $\langle L, R \rangle$ denote a $k$-constrained valid partition of $V(G)$ (if one exists).  We begin by guessing the ``nature of the intersection'' (see below for the exhaustive cases) between $X,  Y$ and $L, R$, and then proceed to finding a $k$-constrained valid partition that respects this guess. 

Note that if $X \cup Y$ is an empty set,  then we can solve the problem by guessing the partition of
$Z$ in time $\calO^*(2^{|Z|})$.
Hence,  we consider the case when $X \cup Y \neq \emptyset$.
Without loss of generality, we can assume that $Y \neq \emptyset$. Consider the following cases.

\begin{enumerate}
\item $X \cap L=\emptyset$ and $X \cap R=\emptyset$. The two sub-cases are as follows.
\begin{enumerate}
\item \label{item:X-empty-Y-intersects-L}
Either $Y \cap L \neq \emptyset$ or $Y \cap R \neq \emptyset$ but not both. 
\item \label{item:X-empty-Y-intersects-both-LR}
$Y \cap L \neq \emptyset$ and $Y \cap R \neq \emptyset$.
\end{enumerate}
\item Either $X \cap L \neq \emptyset$ or $X \cap R \neq \emptyset$ but not both. The two sub-cases are as follows.
\begin{enumerate}
\item \label{item:X-intersects-L-Y-intersects-L}
Either $Y \cap L \neq \emptyset$ or $Y \cap R \neq \emptyset$ but not both. 
\item \label{item:X-intersects-L-Y-intersects-both-LR}
$Y \cap L \neq \emptyset$ and $Y \cap R \neq \emptyset$.
\end{enumerate}

\item $X \cap L \neq \emptyset$ and $X \cap R \neq \emptyset$.
The two sub-cases are as follows.
\begin{enumerate}
\item \label{item:X-intersects-both-LR-Y-intersects-L}
Either $Y \cap L \neq \emptyset$ or $Y \cap R \neq \emptyset$ but not both. 
\item \label{item:X-intersects-both-LR-Y-intersects-both-LR}
$Y \cap L \neq \emptyset$ and $Y \cap R \neq \emptyset$.
\end{enumerate}
\end{enumerate}

\noindent \textbf{Solving Case~(\ref{item:X-empty-Y-intersects-L})}: Without loss of generality let $Y \cap L = \emptyset$. 
Then, for each ordered partition $\langle Z_L, Z_R \rangle$ of $Z$ 
where $L \cap Z=Z_L$ and $R \cap Z=Z_R$, we determine if $\langle Z_L, Z_R \cup Y \rangle$ is a $k$-constrained valid partition.
We declare that $G$ is $k$-contractible to a biclique if and only if
some choice of $\langle Z_L, Z_R \rangle$ leads to a $k$-constrained valid partition. 
The total running time of the algorithm in this case is $\mathcal{O}^*(2^{|Z|})$.

\noindent \textbf{Solving Case~(\ref{item:X-empty-Y-intersects-both-LR})}: 
Consider an ordered partition $\langle Z_L, Z_R \rangle$ of $Z$ 
where $L \cap Z=Z_L$ and $R \cap Z=Z_R$. 
Observe that as $X=\emptyset$, $G$ is connected and 
$Y$ is an independent set, every vertex in $Y$ is adjacent 
to some vertex in $Z_L$ or $Z_R$. 
Next, we apply the following branching rule repeatedly 
to vertices in $Y$ as long as possible. 

\begin{branching rule}
\label{br:X-Y-branch}
If there is a vertex $v \in Y$ such that $N(v) \cap Z_L  \neq \emptyset$, $N(v) \cap Z_R  \neq \emptyset$ and $|N(v) \cap (Z_L \cup Z_R)|>2$, then branch into the following.
\begin{itemize}
\item Contract all edges in $E(v,Z_L)$ and decrease $k$ by the number of contractions.
\item Contract all edges in $E(v,Z_R)$ and decrease $k$ by the number of contractions.
\end{itemize}
\end{branching rule}

The exhaustiveness (and hence the correctness) of Branching Rule \ref{br:X-Y-branch} is easy to verify. 
Also, observe that the branching factor leading to worst-case running time on applying this rule is $(1,2)$. 
After an exhaustive application of Branching Rule \ref{br:X-Y-branch}, any vertex in $Y$ that is adjacent to both $Z_L$ and $Z_R$ has degree 2.   
Then, we apply the following reduction rule.

\begin{preprocessing rule}
\label{rr:X-Y-reduce}
If there is a vertex $v \in Y$ of degree 2 such that $N(v) \cap Z_L  \neq \emptyset$ and $N(v) \cap Z_R  \neq \emptyset$, then contract edges in $E(v,Z_L)$ and decrease $k$ by the number of contractions.
\end{preprocessing rule}

To argue the correctness of this rule, it suffices to show that if $v \notin L$, then $\langle L \cup \{v\}, R \setminus \{v\} \rangle$ is also a $k$-constrained valid partition. 
Let $a \in Z_L$ and $b \in Z_R$ be the neighbours of $v$. 
Let $C_a$ and $C_b$ denote the components of $G[L]$ and $G[R]$ containing $a$ and $b$, respectively.  
Note that $v$ is a pendant vertex in $C_b$. 
The only pair of components of $G[L]$ and $G[R]$ that are adjacent, possibly because of $v$ being in $R$ are $C_a$ and $C_b$. 
Moving $v$ from $C_b$ to $C_a$ does not affect this adjacency or the connectedness of these components since $v$ has exactly 1 neighbour in both $V(C_b)$ and $V(C_a)$. 
The sizes of the spanning forests of $G[L]$ and $G[R]$ before and after the movement of $v$ are equal. 
Hence, $\langle L \cup \{v\}, R \setminus \{v\} \rangle$ is also a $k$-constrained valid partition. 

When neither Branching Rule \ref{br:X-Y-branch} nor Preprocessing Rule \ref{rr:X-Y-reduce} are applicable, we have the partition $\langle Y_L, Y_R \rangle$ of $Y$ where $Y_L=\{y \in Y \mid N(y) \subseteq Z_L\}$ and $Y_R=\{y \in Y \mid N(y) \subseteq Z_R\}$. Let us first consider the case when no vertex in $Y_L \cup Y_R$ is a trivial component of $G[L]$ or $G[R]$. Then, observe that every vertex in $Y_L$ has to be in the same part as $Z_L$ since $N(Y_L) \subseteq Z_L$.
Similarly, every vertex in $Y_R$ has to be in the same part as $Z_R$ since $N(Y_R) \subseteq Z_R$. 
In this case, we simply check if $(Z_L \cup Y_L, Z_R \cup Y_R)$ is a $k$-constrained valid partition. 
The total running time of the algorithm is $\mathcal{O}^*(2^{|Z|}\})$.
Subsequently, we assume that there is a vertex in $Y_L$ or $Y_R$ that is a trivial component of $G[L]$ or $G[R]$. Without loss of generality, let there be a vertex in $Y_L$ that is a trivial component of $G[L]$ or $G[R]$. 
A vertex $y$ of $Y_L$ that is a trivial component of $G[L]$ or $G[R]$ has to be in the same part as $Z_R$ since $y$ has neighbours in $Z_L$. 
Further, no vertex $y'$ in $Y_R$ can be a trivial component of $G[L]$ or $G[R]$, and such a vertex $y'$ has to be in the same part as $Z_R$.  
This is because if $y'$ is in $L$, then the component of $G[L]$ containing $y'$ is not adjacent to the component of $G[R]$ that contains $y$. 
On the other hand, if $y'$ is in $R$, then as $y'$ is adjacent to $Z_R$, it cannot be a trivial component of $G[R]$. 
Now, it follows that every vertex in $Y_R$ is in the same part as $Z_R$, as no vertex in $Y_R$ has neighbours in $Z_L$. 
Therefore, we contract all edges in $E(Y_R, Z_R)$ and decrease $k$ by the number of contractions. Next, we guess a subset $Z' \subseteq Z_L$ that will be in a component of $G[L]$ that has no other vertices from $Z_L$.
There are $\mathcal{O}(2^{|Z|})$ such subsets.
Note that if $Z'$ is empty, then no two vertices of $Z_L$ are in the same component of $G[L]$ or $G[R]$.
Consider the partition of $Y_L$ into three parts $Y^1,Y^2,Y^3$ where $Y^1=\{y \in Y \mid N(y) \subseteq Z'\}$, $Y^2=\{y \in Y \mid N(y) \cap Z' =\emptyset\}$ and $Y^3=Y \setminus (Y^1 \cup Y^2)$.
Then, $Y^1 \subseteq L$ as no vertex in $Y^1$ is a trivial component and $N(Y^1) \subseteq Z'$.
Also, $Y^2 \subseteq L$ as no vertex in $Y^2$ has a neighbour in $Z_R$ or a neighbour in $Z'$.  
Similarly, $Y^3 \subseteq R$ as every vertex in $Y^3$ has a neighbour in $Z'$ and a neighbour outside $Z'$.
Therefore, we simply check if $(Z_L \cup Y^1 \cup Y^2, Z_R \cup Y^3)$ is a $k$-constrained valid partition.  
The total running time of the algorithm in this case is $\mathcal{O}^*(1.619^k 2^{2|Z|})$.

\noindent \textbf{Solving Case~(\ref{item:X-intersects-L-Y-intersects-L})}:  
Without loss of generality, let $X \cap R=\emptyset$. 
Suppose $ Y\cap L=\emptyset$.
Then, for each ordered partition $\langle Z_L, Z_R \rangle$ of $Z$ 
where $L \cap Z=Z_L$ and $R \cap Z=Z_R$, we simply determine 
if $\langle Z_L \cup X, Z_R \cup Y \rangle$ is a $k$-constrained valid partition. 
The total running time in this case is $\mathcal{O}^*(2^{|Z|})$.
Now,  suppose $ Y\cap R=\emptyset$.
Then, for each ordered partition $\langle Z_L, Z_R \rangle$ of $Z$
where $L \cap Z=Z_L$ and $R \cap Z=Z_R$, we simply 
determine if $\langle Z_L \cup X \cup Y, Z_R \rangle$ is a $k$-constrained valid partition.
The total running time in this case is $\mathcal{O}^*(2^{|Z|})$.

\noindent \textbf{Solving Case~(\ref{item:X-intersects-L-Y-intersects-both-LR})}: Without loss of generality let $X \cap R=\emptyset$. 
Consider an ordered partition $\langle Z_L, Z_R \rangle$ of $Z$ where 
$L \cap Z=Z_L$ and $R \cap Z=Z_R$. 
Guess a vertex $y \in Y \cap R$. 
There are at most $n$ choices for $y$. 
Contract $E(y, X) \cup E(y, Z_L)$ and decrease $k$ by the number of contractions. 
Observe that now we have $|X|=1$. 
Then, this case is similar to Case 1(b) after moving the vertex in $X$ to $Z$. 
The total running time in this case is $\mathcal{O}^*(1.619^k \cdot 2^{2|Z|})$.

\noindent \textbf{Solving Case~(\ref{item:X-intersects-both-LR-Y-intersects-L})}: This is similar to Case~(\ref{item:X-intersects-L-Y-intersects-both-LR}).

\noindent \textbf{Solving Case~(\ref{item:X-intersects-both-LR-Y-intersects-both-LR})}: 
In this case, all edges in $E(X,Y)$ except one get contracted. 
Hence, if $|X \cup Y|>k+2$, we declare that $(G,k)$ is a \no-instance. 
Otherwise, $|V(G)| \leq |Z|+k+2$. 
Then, we go over each ordered partition of $V(G)$ into two parts $\langle Q_L, Q_R \rangle$ and determine if $\langle Q_L, Q_R \rangle$ is a $k$-constrained valid partition. 
The running time in this case is $\mathcal{O}^*(2^{|Z|+k})$.

Now, we have the following result. 

\begin{lemma}
Given a graph $G$ and a biclique modulator $Z$, there is an algorithm that determines if $G$ is $k$-contractible to a biclique or not in $\mathcal{O}^*(4^{|Z|} 1.619^k)$ time.
\end{lemma}

As $Z$ can be obtained in $\mathcal{O}^*(1.4^{|Z|})$ time where $|Z| \leq 2k$, the claimed running time of $\mathcal{O}^*(25.904^k)$ to solve \textsc{Biclique Contraction} follows establishing part 1 of Theorem \ref{thm:fpt-algo}. 

The only changes required in the algorithm for \textsc{Balanced Biclique Contraction} are in the places where a partition of the graph into two parts is verified if it is a $k$-constrained valid partition or not. 
This check has to be modified such that the verification is made to determine if the partition is a $k$-constrained valid balanced partition or not. 
The remaining parts of the algorithm remain as such. 
This establishes part 2 of Theorem \ref{thm:fpt-algo}. 
\section{Kernelization Complexity} 
\label{sec:ker}

We begin by observing that the reduction from \textsc{Red-Blue Dominating Set} to \textsc{Biclique Contraction} described in Section \ref{subsec:np-hard-biclique} is a polynomial parameter transformation as it maps an instance $(G, R, B, \kappa)$ to $(H,k=|B|+\kappa)$ where $|B|+\kappa$ is the parameter of the input instance and $k$ is the parameter of the output instance. Using the incompressibility and infeasibility of polynomial-size kernels (with respect to $|B|+\kappa$ as the parameter) result known for \textsc{Red-Blue Dominating Set} \cite{CyganFKLMPPS15, DomLS14}, we obtain part 2 of Theorem \ref{thm:kernels}. In fact, we can also conclude that {\sc Biclique Contraction} does not even admit a polynomial lossy kernel unless \NP\ $\subseteq$ \coNP/poly \cite{LokshtanovPRS17}. 

We proceed to describe a quadratic vertex kernel for \textsc{Balanced Biclique Contraction} using a sequence of reduction rules. 
The reduction rules are ordered (in the sequence stated) and a rule is applied on the instance only when none of the earlier reduction rules are applicable. Consider an instance $(G, k)$. Let $\mathcal{C}$ be a maximal collection of vertex-disjoint $K_1+K_2$ and $K_3$ in $G$. 
Let $Z= \bigcup_{C \in \mathcal{C}} V(C)$.  

\begin{reduction rule}
\label{rr:trivial-ker}
If $k \leq 0$ and $G$ is not a balanced biclique or if $|Z|>6k$, then return a trivial  \no-instance. 
\end{reduction rule}

The correctness of the first part of the rule is easy to verify. 
Consider the second part.  
Suppose $(G,k)$ is a \yes-instance. 
Then, from Observation \ref{obs:bc-vertex-sol}, there is a set $\widehat{Z}$ of at most $2k$ vertices such that $G-\widehat{Z}$ is a biclique. 
Further, from Lemma \ref{lemma:biclique-forbidden-char}, $G-\widehat{Z}$ is $\{K_1+K_2,K_3\}$-free. 
As $\widehat{Z} \cap V(S) \neq \emptyset$ for each $S \in \mathcal{C}$, it follows that $|\mathcal{C}| \leq 2k$. 
Then, this implies that $|Z| \leq 6k$.

Let $\langle X, Y \rangle$ be a bipartition of $G - Z$ where $|X| \le |Y|$. 
Subsequently, we assume that $|Y| \geq k+3$. 
Otherwise, we have a linear vertex kernel. 

\begin{reduction rule}
\label{rr:X-Y-sizes}
If $|Y| > |X| + |Z| + k$, then return a trivial \no-instance.
\end{reduction rule}

\begin{lemma}
\label{lem:Y-not-so-large-X}
Reduction Rule \ref{rr:X-Y-sizes} is safe.
\end{lemma}
\begin{proof}
We will show that if $(G, k)$ is a \yes-instance, then $|Y| \le |X| + |Z| + k$. Suppose $G$ is $k$-contractible to the balanced biclique $G/F$ and $\calW$ is the $G/F$-witness structure of $G$.
Let $\calW_L$ and $\calW_R$ be the collection of witness sets 
corresponding to the bipartition $\langle L, R \rangle$ of $G/F$. 
Note that $|\calW_L| = |\calW_R|$. 
As $|F| \le k$ and $Y$ is an independent set in $G$, there are at most $k$ vertices in $Y$ that are incident on some edges in $F$.
Let $Y'$ be the collection of all such vertices.
As $|Y| \ge k + 1$, $Y \setminus Y'$ is a non-empty set.
Note that every vertex in $Y \setminus Y'$ is in a singleton witness set in $\calW$. 
As there is no edge between any of these singleton witness sets, all of these vertices are either in $\calW_L$ or in $\calW_R$. Without loss of generality,  let every vertex in $Y \setminus Y'$ be in a witness set in $\calW_R$.
Then, $|\calW_R| \ge |Y| - |Y'| \ge |Y| - k$. 
Further, there is no singleton witness set in $\calW_L$ that contains a vertex in $Y$ since such a witness set cannot be adjacent to a singleton witness set in $\calW_R$ that has a vertex in $Y \setminus Y'$.
This implies that every witness set in $\calW_L$ contains at least one vertex from $X \cup Z$.
Hence, $|X| + |Z| \ge |\calW_L|$.
However, $|\calW_L| = |\calW_R| \ge |Y| - k$ which implies $|X| + |Z| \ge |Y| - k$.
\end{proof}

Now, we have the following important property on $G$: if $(G,k)$ is a \yes-instance then, in any $k$-constrained valid balanced partition $\langle L,R \rangle$ of $V(G)$ either $X \subseteq L$, $Y \subseteq  R$ or $X \subseteq R$, $Y \subseteq  L$. 
Suppose there are vertices $x_1,x_2 \in X$ such that $x_1 \in L$ and $x_2 \in R$. 
If $Y \subseteq L$, then since $|Y| \geq k+1$, the size of a spanning forest of $G[L]$ exceeds $k$. 
Similarly, if $Y \subseteq R$, then the size of a spanning forest of $G[R]$ exceeds $k$. 
Therefore, there are vertices $y_1,y_2 \in Y$ such that $y_1 \in L$ and $y_2 \in R$.
Let $|X \cap L|=\alpha_X$, $|X \cap R|=\beta_X$, $|Y \cap L|=\alpha_Y$ and $|Y \cap R|=\beta_Y$.
Then, any spanning forest of $G[L]$ has $k_\ell \geq \alpha_X + \alpha_Y-1$ edges and any spanning forest of $G[R]$ has $k_r \geq \beta_X + \beta_Y-1$ edges with $k_\ell+k_r>k$. 

Let $Z_X=\{z \in Z : |N(z) \cap Y| \ge k + 1 \}$, $Z_Y=\{z \in Z : |N(z) \cap X| \ge k + 1 \}$ and $Z'=Z \setminus (Z_X \cup Z_Y)$. 
Observe that if $Z_X \cap Z_Y \neq \emptyset$, then $(G,k)$ is a \no-instance. 
Suppose there is a vertex $z \in Z_X \cap Z_Y$ and $(G,k)$ is a \yes-instance with $\langle L,R \rangle$ being a $k$-constrained valid balanced partition of $V(G)$. 
Without loss of generality, let $X \subseteq L$, $Y \subseteq  R$. 
However, as $z$ has at least $k+1$ neighbours each in $X$ and $Y$, $z$ cannot be in the part containing $X$ or in the part containing $Y$, leading to a contradiction. 
Hence, we may assume that the sets $Z'$, $Z_X$, and $Z_Y$ partition $Z$. The next reduction rule is a simplification rule based on this partition.

\begin{reduction rule}
\label{rr:contract-Z-X-Y}
If there is an edge in $E(X,Z_X) \cup E(Y,Z_Y) \cup E(Z_X) \cup E(Z_Y)$, then contract it and decrease $k$ by 1.
\end{reduction rule}

During the contraction of an edge incident on a vertex $v$ in $X \cup Y$ in the application of Reduction Rule \ref{rr:contract-Z-X-Y}, the new vertex added in the process is renamed as $v$ and retained in $X \cup Y$. 
Observe that the resulting graph $G[X \cup Y]$ is also a biclique.  

\begin{lemma}
Reduction Rule \ref{rr:contract-Z-X-Y} is safe. 
\end{lemma}
\begin{proof}
Suppose $(G,k)$ is a \yes-instance and $\langle L,R \rangle$ is a $k$-constrained valid balanced partition of $V(G)$. 
Without loss of generality, $X \subseteq L$, $Y \subseteq  R$. 
Since any $z \in Z_X$ has at least $k+1$ neighbours in $Y$, if $z \in R$, then any spanning forest of $G[R]$ has at least $k+1$ edges, leading to a contradiction. 
Similarly, any $z \in Z_Y$ has at least $k+1$ neighbours in $X$, and if $z \in L$, then any spanning forest of $G[L]$ has at least $k+1$ edges, leading to a contradiction. 
Therefore, $Z_X \subseteq L$ and $Z_Y \subseteq R$. 
This justifies contracting an edge in $E(X,Z_X) \cup E(Y,Z_Y) \cup E(Z_X) \cup E(Z_Y)$.
\end{proof}

Now, we have the following property on $G$: if $(G,k)$ is a \yes-instance then, in any $k$-constrained valid balanced partition $\langle L, R \rangle$ of $V(G)$, we have $X \cup Z_X \subseteq L$ and $Y \cup Z_Y \subseteq  R$. 
The final reduction rule marks certain essential vertices and deletes the non-essential ones.

\begin{reduction rule}
\label{rr:unmark}
Mark all vertices in $Z$, $N(Z') \cap X$ and $N(Z') \cap Y$. Further, for each vertex $z \in Z$, mark one of its non-neighbours (if it exists) each in $X \setminus N(Z')$ and $Y \setminus N(Z')$. After performing this marking scheme, if there are at least two unmarked vertices in $X$ and at least two unmarked vertices in $Y$, then delete two unmarked vertices $u \in X$ and $v \in Y$.
\end{reduction rule}

The correctness of this rule is justified by the following lemma.

\begin{lemma}
\label{lem:rr-unmark}
Reduction Rule~\ref{rr:unmark} is safe.
\end{lemma}
\begin{proof}
Suppose $G$ is $k$-contractible to a balanced biclique. 
Consider a $k$-constrained valid balanced partition $\langle L,R \rangle$ of $V(G)$. 
Then, we know that $X \cup Z_X \subseteq L$ and $Y \cup Z_Y \subseteq R$.  
As $N(u) \subseteq Y \cup Z_Y$ and $N(v) \subseteq X \cup Z_X$, it follows that $u$ and $v$ are in trivial components of $G[L]$ and $G[R]$, respectively. 
Then, $\langle L \setminus \{u\}, R \setminus \{v\}\rangle$ is a $k$-constrained valid balanced partition of $V(G) \setminus \{u,v\}$.

Conversely, consider a $k$-constrained valid balanced partition $\langle L',R'\rangle$ of $V(G) \setminus \{u,v\}$.  
Let $G'$ denote $G-\{u,v\}$. 
Observe that $X \setminus \{u\} \subseteq L'$ and $Y \setminus \{v\} \subseteq R'$. 
We show that $\langle L' \cup \{u\}, R' \cup \{v\}\rangle$ is a $k$-constrained valid balanced partition of $V(G)$. 
Suppose $u$ is in a trivial component of $G[L' \cup \{u\}]$ and is not adjacent to some component $H$ of $G[R' \cup \{v\}]$. 
Then, $V(H) \subseteq Z' \cup Z_Y$ and let $z \in V(H) $. 
Let $u^*$ be a non-neighbour of $z$ in $X \setminus N(Z')$ that was marked. 
Since $N(u^*) \subseteq Y \cup Z_Y$, $u^*$ is in a trivial component of $G'[L']$. 
In order for the component of $G'[R']$ containing $z$ to be adjacent to the component of $G'[L']$ containing $u^*$, $z$ must be in a component that contains some vertex in $Y$. 
Then, since $Y \subseteq N(u)$, this leads to a contradiction.
A symmetric argument holds when $v$ is in a trivial component of $G[R' \cup \{v\}]$ and is not adjacent to some component of $G[L' \cup \{u\}]$. 
Therefore, $\langle L' \cup \{u\}, R' \cup \{v\}\rangle$ is a $k$-constrained valid balanced partition of $V(G)$. 
\end{proof}

Observe that after the application of Reduction Rule~\ref{rr:unmark}, there are $\mathcal{O}(k)$ unmarked vertices. 
Further, the marking rule procedure marks $\mathcal{O}(k^2)$ vertices. This establishes Theorem \ref{thm:kernels}.

\section{$K_{p,*}$-Contraction}
Contracting a graph to $K_{p, q}$ for fixed values of $p,q$ is well-studied~\cite{BrouwerV87}. In particular, contracting a graph to $K_{1, 1}$ or $K_{1, 2}$ is polynomial-time solvable while contracting a graph to $K_{2, 2}$ is \NP-complete. In this section, we study the complexity of the $K_{p,*}$-\textsc{Contraction} problem of contracting a graph to a biclique with exactly $p$ vertices in one part using at most $k$ edge contractions. 


\subsection{\NP-Completeness}
It is easy to verify that $K_{p, *}$-\textsc{Contraction} is in \NP\ for each $p \in \mathbb{N}^+$. 
The \NP-completeness of $K_{1, *}$-\textsc{Contraction} is already known from the work of Heggernes et al.~\cite{HeggernesHLLP14}. We show the \NP-hardness of $K_{p, *}$-\textsc{Contraction} for each $p \geq 2$ by giving a polynomial-time reduction from \textsc{Positive Not-All-Equal SAT (PNAESAT)}, thus establishing Theorem \ref{thm:np-hard-kcq-correctness}. In this problem, we are given a Boolean formula $\phi$ in conjunctive normal form in which no clause has negative literals and the objective is to determine if $\phi$ has a satisfying assignment in which every clause has a variable set to \true\ and a variable set to \false. The \NP-hardness of \textsc{PNAESAT} follows from the \NP-hardness of \textsc{Hypergraph $2$-Coloring}.

\begin{lemma}
\label{lem:np-hard-kcq-reduction}
There is a polynomial-time reduction that takes as input an instance $\phi$ of \textsc{PNAESAT} and returns an equivalent instance $(G, k)$ of $K_{p, *}$-\textsc{Contraction} such that $|V(G)| = \mathcal{O}(nm)$ and $k=3n$, where $n$ is the number of boolean variables and $m$ is the number of clauses in the input instance $\phi$. 
\end{lemma}
\begin{proof}
Consider an instance $\phi$ of \textsc{PNAESAT} with $V = \{x_1, x_2 \ldots x_n \}$ denoting the set of variables and $C= \{c_1,c_2, \ldots c_m\}$ denoting the collection of clauses. We construct an instance $(G, k=3n)$ of $K_{p, *}$-\textsc{Contraction} where $G$ is obtained from $\phi$ as follows.
\begin{itemize}
   \item Add two vertices $t$, $f$ and a set $M$ of $6n+p+1$ vertices  with every vertex in $M$ adjacent to both $t$ and $f$. 
    \item For each variable $x_i \in V$, add three vertices $x_i,x^t_i,x^f_i$ such that  $x^t_i$ is adjacent to $t$ and $x_i$ while $x^f_i$ is adjacent to $f$ and $x_i$. Let $X_V=\{x_i : i \in [n]\}$ and $Z=\{x^t_i,x^f_i : i \in [n]\}$. 
    \item Add a vertex $r$ adjacent to $t$, $f$ and $x^t_i$ and $x^f_i$ for each $i \in [n]$. 
    \item For each clause $c_i \in C$, add $6n+1$ vertices $c_i^1,c_i^2,\ldots,c_i^{6n+1}$. Let $C'$ denote the set of such vertices and let $C^j=\{c_i^j: i \in [m]\}$ for each $j \in [6n+1]$. 
    \item For each variable $x_p$ and clause $c_i$ such that $c_i$ has $x_p$, add an edge between $x_p$  and $c_i^j$ for every $j \in [6n+1]$.
    \item Add a set $L$ of $p-2$ vertices adjacent to  every vertex in $\{r\} \cup C' \cup M$.
\end{itemize}

\begin{figure}[t]
\begin{center}
\includegraphics[scale=0.3]{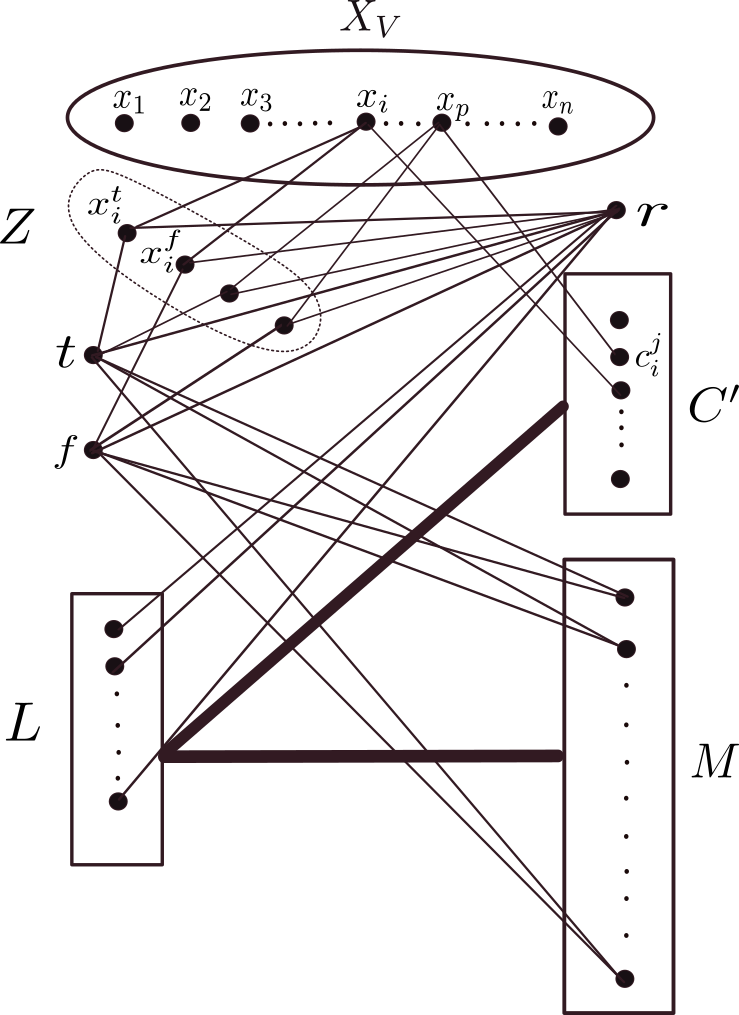}
\caption{The graph $G$ corresponding to the formula $\phi$ in the reduction from \textsc{PNAESAT} to $K_{p, *}$-\textsc{Contraction}. A thick line between two sets of vertices denotes that every vertex of one set is adjacent to every vertex of the other set. 
\label{fig:np-hard-kcqbiclique}}
\end{center}
\end{figure}

\noindent See Figure \ref{fig:np-hard-kcqbiclique} for an illustration. It is easy to verify that the reduction takes polynomial time and $G$ has $\mathcal{O}(mn)$ vertices. Now we proceed to the correctness.

$(\Rightarrow)$ Suppose $\alpha: V \rightarrow \{T, F \}$ is a satisfying assignment of $\phi$ such that every clause has a variable set to \true\ and a variable set to \false. We partition $V(G)$ into two parts $X$ and $Y$ as follows. 
\begin{itemize}
    \item $\{t,f\} \cup L \subseteq X$ and $C' \cup M \cup \{r\} \subseteq Y$.
    \item For each $i \in [n]$ such that $\alpha$ sets $x_i$ to \true,  add $x_i$, $x^t_i$ to $X$ and $x^f_i$ to $Y$.
    \item For each $i \in [n]$ such that $\alpha$ sets $x_i$ to \false,  add $x_i$, $x^f_i$ to $X$ and $x^t_i$ to $Y$.
\end{itemize}

\noindent Observe that $t$ and $f$ are in different components of $G[X]$. Further, each vertex of $L$ is a trivial component of $G[X]$. Thus, $G[X]$ has exactly $p$ components. All vertices in $Z \cap Y$ are in a single component of $G[Y]$ along with $r$. Every other component of $G[Y]$ is a trivial component. It is easy to verify that every component of $G[X]$ is adjacent to every component of $G[Y]$. Further, the number of edges in a spanning forest of $G[X]$ is $2n$ and the number of edges in the spanning forest of $G[Y]$ is $n$. Hence, $\langle X,Y \rangle$ is a $k$-constrained valid partition of $G$ with $G[X]$ containing exactly $p$ components.

$(\Leftarrow)$
Conversely, suppose $(G, k)$ is an \yes-instance of $K_{p, *}$-\textsc{Contraction}. 
Let $\langle L^*,R^* \rangle$ be a $k$-constrained valid partition of $V(G)$ given by Lemma~\ref{lem:bc-part}. Without loss of generality, let $G[L^*]$ have exactly $p$ components. As $|M| = 6n+p+1$ and there are at most $2k=6n$ vertices in big witness sets of the corresponding biclique witness structure, it follows that at least $p+1>p \geq 2$ vertices of $M$ are in trivial components of $G[R^*]$. As $t$ and $f$ are adjacent to all vertices in $M$, it follows that $t, f$ are in $L^*$. Each component of $G[R^*]$ (including the trivial components) is adjacent to every component of $G[L^*]$. We know that $|L|=p-2$ and the neighbourhood of each vertex in $M$ is $L \cup \{t,f\}$. Since $G[L^*]$ has exactly $p$ components, it follows that $L \subseteq L^*$ (in order for the trivial components of $G[R^*]$ consisting of vertices from $M$ to be adjacent to components of $G[L^*]$). Further, no two vertices of $L$ are in the same component of $G[L^*]$. Moreover, $t$ and $f$ are in two different components ($C_t$ and $C_f$, respectively) of $G[L^*]$. This implies that $r \in R^*$ and $M \subseteq R^*$. 

We claim that no vertex from $X_V$ is in $R^*$. Recall that each vertex $x_i \in X_V$ is adjacent to at least $6n+1$ vertices of $C^\prime$. As $|L^*| \leq 3n+p$ and $L \cup \{t,f\} \subseteq L^*$, at most $3n$ vertices of $C^\prime$ can be in $L^*$. If there is a vertex $x_i \in X_V \cap R^*$, then it is in the same component of $G[R^*]$ that contains at least $3n+1$ of its neighbours. This implies that $\langle L^*, R^* \rangle$ is not a $k$-constrained valid partition of $V(G)$ leading to a contradiction. Therefore $X_V \subseteq L^*$. Next, we claim that no vertex from $C^\prime$ is in $L^*$. Suppose $c^j_i$ is in $L^*$.  Then all vertices in $X_V$ corresponding to variables contained in clause $c_i$ are in the same component of $G[L^*]$ as $c^j_i$, say $C$. Then $G[L^*]$ has a component $\widehat{C} \neq C$ such that $V(\widehat{C}) \cap L=\emptyset$. As $|L^*| \leq 3n+p$, not every vertex in $\{c^1_i,c^2_i,\ldots,c^{6n+1}_i\}$ is in $L^*$. The vertices from this set that are in $R^*$ (as trivial components of $G[R^*]$) are not adjacent to $\widehat{C}$, leading to a contradiction. Therefore, $C^\prime \subseteq R^*$. As $\{r\} \cup C' \cup M \subseteq R^*$, every vertex in $L$ is in a trivial component of $G[L^*]$ and $X_V \subseteq V(C_t) \cup V(C_f)$. As $N(C^\prime) =X_V \cup L $, every vertex in $C'$ is in a trivial component of $G[R^*]$. Consider a vertex $c^j_i \in C'$. It follows that $C_t$ and $C_f$ each contain a vertex $x'$ and $x''$ such that the clause $c_i$ in $\phi$ contains the variables $x'$ and $x''$. Consider a truth assignment $\alpha$ for $\phi$ defined as follows.
\begin{itemize}
    \item If $x_i \in V(C_t)$ then set the variable $x_i \in V$ to \true.
    \item If $x_i \in V(C_f)$ then set the variable $x_i \in V$ to \false. 
\end{itemize}
Clearly, $\alpha$ is a satisfying assignment for $\phi$ in which each clause has a variable set to \true\ and a variable set to \false.
\end{proof}

\subsection{Fixed-Parameter Tractability}

We begin with the following result on $K_{1, *}$-\textsc{Contraction}. Observe that the objective is to contract the input graph to a star.

\begin{lemma}
\label{lemma:fpt-algo-1-witness-set-in-left}
There is an algorithm that given a graph $G$ and an integer $k$, determines whether $G$ is $k$-contractible to $K_{1, *}$ in $\mathcal{O}^*(2^k)$ time.
\end{lemma}
\begin{proof}
From Lemma \ref{lem:bc-part}, determining if $G$ is $k$-contractible to a star is equivalent to determining if $V(G)$ has a $k$-constrained valid partition $\langle L, R \rangle$ such that the number of components of $G[L]$ is one. 
We further claim that there is such a partition $\langle L, R \rangle$ where each component of $G[R]$ is trivial. That is, $L$ is a connected vertex cover of $G$ and $|L| \leq k+1$. 
Assume, on the contrary, that $G[R]$ has a non-trivial component $C$. Then, there are vertices $x \in L$ and $y \in V(C)$ such that $xy \in E(G)$. Consider a spanning forest of $C$ and root it at $y$. Let $z$ be one of the leaves of this forest. Now,  $\langle L \cup (V(C) \setminus \{z\}), R \setminus (V(C) \setminus \{z\}) \rangle$ is another $k$-constrained valid partition of $V(G)$. By repeating this procedure, we arrive at the required partition. 
Therefore, determining if $G$ is $k$-contractible to a star is equivalent to determining if $G$ has a connected vertex cover of size at most $k+1$. Using the current best \FPT\ algorithm for \textsc{Connected Vertex Cover} \cite{Cygan12}, it follows that this problem can be solved in $\mathcal{O}^*(2^k)$ time.
\end{proof}

Lemma \ref{lemma:fpt-algo-1-witness-set-in-left} proves the first part of Theorem \ref{thm:fpt-algo-c-witness-set-in-left}. Now we show an \FPT\ algorithm for $K_{p, *}$-\textsc{Contraction} for each $p \geq 2$.

\begin{lemma}
\label{thm:fpt-algo-c-witness-sets-in-left}
 For each $p \geq 2$, $K_{p, *}$-\textsc{Contraction} can be solved in $\mathcal{O}^*(3^k)$ time. 
\end{lemma}
\begin{proof}
Consider an instance $(G,k)$ of $K_{p, *}$-\textsc{Contraction}. From Lemma \ref{lem:bc-part}, determining if $G$ is $k$-contractible
to $K_{p, *}$ is equivalent to determining 
if $V(G)$ has a $k$-constrained valid partition $\langle L,R \rangle$ 
such that the number of components of $G[L]$ is $p$. We will now describe an algorithm that determines the existence
of a partition $\langle L, R \rangle$ of $V(G)$ with these properties. 

We first select $p$ vertices from $V(G)$ with the interpretation that each of them is in a different component of $G[L]$. There are at most $n^p$ choices for selecting this set. Consider one such choice and let  $\widehat{L}$ denote these $p$ vertices. Let $\widehat{R} = \emptyset$. If $(G,k)$ is a \yes-instance, then $\widehat{L}$ and $\widehat{R}$ will grow to be the required sets $L$ and $R$ for some choice of the initial $p$ vertices. We proceed by applying the following reduction rule to $V(G)\setminus \widehat{L}$ repeatedly.

    
\begin{reduction rule}
\label{rr:maintain-c-components}
If there is a vertex $v \in V(G)\setminus \{\widehat{L} \cup \widehat{R} \}$ such that $v$ is adjacent to two or more connected components of $G[\widehat{L}]$, then add $v$ to $\widehat{R}$.
\end{reduction rule}

The correctness of this rule is easy to verify given our interpretation of the initial $p$ vertices. During the entire course of the algorithm, at any point, if this rule is applicable, we apply it. Subsequently let $U$ = $V(G)\setminus \{\widehat{L} \cup \widehat{R} \}$. Now we apply the following branching rule to vertices in $U$.

\begin{branching rule}
\label{br:2-noureighb-branch}
If there is a vertex $v \in U$ such that $N(v) \cap \widehat{L}  \neq \emptyset$ and $N(v) \cap \widehat{R} \neq \emptyset$, then branch into the following.
\begin{itemize}
\item Contract all edges in $E(v,\widehat{L} )$ and decrease $k$ by the number of contractions. 
\item Contract all edges in $E(v,\widehat{R} )$ and decrease $k$ by the number of contractions.
\end{itemize}
\end{branching rule}
The exhaustiveness (and hence the correctness) of Branching Rule \ref{br:2-noureighb-branch} follows from the fact that each $v \in V(G)$ is in $L$ or in $R$. Also, observe that the worst-case branching factor is $(1,1)$ for this rule. After an exhaustive application of this rule, any vertex in $U$ is adjacent to either $\widehat{L}$ or $\widehat{R}$. 
The total number of leaves of the search tree is $\mathcal{O}({2^{k^{\prime}}})$ where $k^{\prime} \leq k$.
At the leaves where $k^{\prime} = k$, we check if $\langle \widehat{L}, V(G) \setminus \widehat{L} \rangle$ is a $k$-constrained valid partition of $V(G)$ such that $G[\widehat{L}]$ has exactly $p$ components and declare that $(G,k)$ is a \yes-instance if it is indeed the case. At the leaves where $k^{\prime} < k$, we apply the following branching rule repeatedly.

\begin{branching rule}
\label{br:edge-U-branch}
If there is an edge $uv$ in $G[U]$, then branch into the following cases. 
\begin{itemize}
\item $\widehat{L} = \widehat{L} \cup \{u\}$.
\item $\widehat{L} = \widehat{L} \cup \{v\}$.
\item Contract the edge $uv$ into a new vertex $x$, add $x$ to $\widehat{R}$ and increase $k^{\prime}$ by one.
\end{itemize}
\end{branching rule}
At any point of time, if $k^{\prime} = k$, we stop branching and check if $\langle \widehat{L}, V(G) \setminus \widehat{L} \rangle$ is a $k$-constrained valid partition of $V(G)$ such that $G[\widehat{L}]$ has exactly $p$ components and declare that $(G,k)$ is a \yes-instance if it is indeed the case. To bound the number of nodes in the search tree, we associate a measure $\mu(N)$ to every node $N$ defined as $\mu(N) = k-{k^\prime}+p-|\widehat{L}|$. In the first and second branches, $|\widehat{L}|$ increases by $1$ and hence $\mu(N)$ drops by $1$. In the third branch, $|\widehat{L}|$ is unchanged but $k^\prime$ increases by $1$ and hence $\mu(N)$ drops by $1$. At any leaf of this search tree, one of the following cases holds.
\begin{itemize}
    \item Case 1:  $|\widehat{L}|=k-{k^\prime}+p$. In this case, we simply check if $\langle \widehat{L}, V(G)\setminus \widehat{L} \rangle$ is the required partition and declare that $(G,k)$ is a \yes-instance if it indeed is. 
    \item Case 2: There is no edge $uv$ in $G[U]$. In this case, $G[U]$ induces an independent see,t and since $G$ is connected, every vertex in $G[U]$ has a neighbour in either $\widehat{L}$ or $\widehat{R}$ but not both (otherwise, Branching Rule \ref{br:2-noureighb-branch} would have been applied). For each $x \in U$ such that $x$ has a neighbour in $\widehat{L}$, move $x$ to $\widehat{L}$. Similarly, for each $x \in U$ such that $x$ has a neighbour in $\widehat{R}$, move $x$ to $\widehat{R}$. The correctness of this step is justified by the fact that any vertex $x \in U$ that has a neighbour in $\widehat{L}$ cannot be in the part $R$ of the desired partition $\langle L, R \rangle$ as it is adjacent to only one component of $G[L]$.  Any vertex $x \in U$ that has a neighbour in $\widehat{R}$ cannot be in the part $L$ of the desired partition $\langle L, R \rangle$ because part $L$ induces only $p$ connected components. Now we just check if $\langle \widehat{L}, \widehat{R} \rangle$ is the required partition and declare that $(G,k)$ is a \yes-instance if it is so.
\end{itemize}
If no leaf of the search tree ascertains that $(G,k)$ is a 
\yes-instance, then we declare that $(G,k)$ is a \no-instance. The total running time of this algorithm is  $\mathcal{O}({2^{k^{\prime}}} \cdot {3^{k-{k^{\prime}}+p-|\widehat{L}|}} \cdot n^p)$ which is $\mathcal{O}^*(3^k)$. 
\end{proof}

\section{Concluding Remarks}
\label{sec:concl}

In this work, we initiated the study of {\sc Biclique Contraction} and {\sc Balanced Biclique Contraction}. 
We showed \NP-completeness, fixed-parameter tractability and kernelization results for the problems. We also gave a faster  algorithm when the resultant biclique is restricted. The parameter of interest in all these results is the solution size, i.e., the number of edge contractions. A natural future direction is to study \textsc{(Balanced) Biclique Contraction} with respect to the size $\ell$ of the target (balanced) biclique as the parameter. Note that the parameterized complexity of \textsc{Balanced Biclique Vertex Deletion} with respect to the number $\ell$ of vertices in the resultant balanced biclique was a long-standing open problem until it was shown to be \W[1]-hard in \cite{Lin18}. However, the simple reduction that takes an instance $(H, k)$ of \textsc{Independent Set} and constructs an instance $(G, \ell)$ of \textsc{Biclique Contraction} where $G$ is obtained by adding a universal vertex to $H$ leads to the following result.

\begin{theorem}
\label{thm:w-hard-target-size}
{\sc Biclique Contraction} parameterized by the lower bound $\ell= n - k$ on the number of vertices in the resultant biclique is \W$[1]$-\hard. 
\end{theorem} 

Observe that if $G$ is contractible to a biclique on $2 \cdot \ell$ vertices, then $G$ is also contractible to a biclique on $\ell$ vertices. Hence,  we can guess $\ell_0 \in \{\ell, \ell + 1, \dots, 2\cdot \ell\}$ where $\ell_0$ is the smallest integer greater than or equal to $\ell$ such that $G$ is contractable to a biclique on $\ell_0$ vertices. However, we could not obtain a simple algorithm running in $n^{f(\ell_0)}$ time for determining if $G$ can be contracted to a biclique on $\ell_0$ vertices. In contrast, if $\ell_0$ is the smallest integer greater than or equal to $\ell$ such that $G$ can be contracted into a balanced biclique on $\ell_0$ vertices, then there is no such easy upper bound on $\ell_0$. Hence, we conjecture  that \textsc{Balanced Biclique Contraction} is \para-\NP-\hard\ when parameterized by the lower bound $\ell= n - k$ on the number of vertices in the resultant balanced biclique.


\end{document}